\documentclass[12pt,draftclsnofoot,onecolumn]{IEEEtran}

\usepackage[numbers,sort&compress,square]{natbib}   
\usepackage{graphicx}
\usepackage{subfigure}
\usepackage{amsmath,amssymb}
\usepackage{amsfonts}
\usepackage{array}
\usepackage{amsthm}
\usepackage{mathrsfs}
\usepackage{amsmath}
\usepackage{setspace}
\usepackage{color}
\usepackage[ruled,linesnumbered]{algorithm2e}
\usepackage{soul}

\newtheorem{theorem}{Theorem}

\newtheorem{lemma}{Lemma}

\theoremstyle{definition}
\newtheorem{define}{Definition}

\newtheorem{remark}{Remark}
\newtheorem{example}{Example}

\renewcommand{\baselinestretch}{1.4}

\begin{document}

\title{Joint Index Coding and Incentive Design for Selfish Clients}

%

\author{Yu-Pin Hsu, I-Hong Hou, and Alex Sprintson
	\footnote{Y.-P. Hsu is with  Department	of Communication Engineering, National Taipei University, Taiwan.  I-H. Hou and A. Sprintson are with Department of Electrical and Computer Engineering, Texas A\&M University, USA.   Email: \texttt{yupinhsu@mail.ntpu.edu.tw, \{ihou, spalex\}@tamu.edu}. 	This paper was presented in part in Proc. of IEEE ISIT \cite{hsu2013index}. This work was supported by Ministry of Science and Technology, Taiwan, under grant MOST 107-2221-E-305-007-MY3. }
}

\maketitle

\vspace{-1cm}

\begin{abstract}	
The  index coding problem includes a server,  a group of clients, and a set of data chunks.  While each client wants a subset of the data chunks and already has another subset as its side information, the server transmits some uncoded data chunks or coded data chunks to the clients over a noiseless broadcast channel. The objective of the problem is to satisfy the demands of all clients with the minimum number of transmissions. In this paper, we investigate the index coding setting from a game-theoretical perspective. We consider \textit{selfish} clients, where each selfish client has \textit{private side information} and a \textit{private valuation} of each data chunk it wants. In this context, our objectives are following: 1) to motivate each selfish client to reveal the correct side information and true valuation of each data chunk it wants; 2)  to maximize the social welfare, i.e., the total valuation of the data chunks recovered by the clients minus the total cost incurred by the transmissions  from the server. Our main contribution is to  jointly develop coding schemes and incentive schemes for achieving the first objective perfectly  and achieving the second objective optimally or approximately with guaranteed approximation ratios (potentially within some restricted sets of coding matrices). 
\end{abstract}

\section{Introduction} \label{section:introduction}

There has been a  dramatic proliferation of research on \textit{wireless network coding} because it can substantially reduce the number of transmissions by the \textit{broadcast} nature of a wireless medium. On one hand, the wireless medium allows a wireless sender node to broadcast data to all neighboring nodes with a single transmission. On the other hand, a wireless receiver node can  overhear the wireless channel and store the overheard data for decoding future transmissions, which  is referred to as \textit{side information}. Take the wireless network in Fig.~\ref{fig:index-coding}-(a) for example, where  sender $s_1$  sends  data chunk $d_1$ to  receiver $r_1$ through a  forwarder and  sender $s_2$ sends data chunk $d_2$ to  receiver $r_2$ also through  the forwarder. While receiver $r_1$ can obtain data chunk $d_2$ destined to $r_2$ by overhearing the transmissions from $s_2$,  receiver $r_2$ can also obtain data chunk $d_1$ destined to $r_1$  by overhearing the transmissions from $s_1$. 
Leveraging the side information, the forwarder can simply broadcast a single XOR-coded data chunk $d_1+d_2$, and then both receivers can obtain the data chunks they want by subtracting the side information they have from the received data chunk $d_1+d_2$. However, with the conventional communication approach (without coding), the forwarder has to transmit both data chunks $d_1$ and $d_2$ separately.



The index coding problem is one of  fundamental problems on the wireless network coding. An instance of the index coding problem includes a server (playing the role of the forwarder in Fig.~\ref{fig:index-coding}-(a)), a set of wireless clients, and a set $\mathbf{D}$ of data chunks. Each client wants a subset of the data chunks in set $\mathbf{D}$ and has a different subset of the data chunks in set $\mathbf{D}$ given to it as side information. The server can transmit uncoded data chunks or coded data chunks (i.e., combinations of  data chunks in set $\mathbf{D}$) to all clients over a noiseless broadcast channel. The goal of the problem is to identify a coding (transmission) scheme  requiring the minimum number  of transmissions to satisfy the demands of all clients. For example, Fig.~\ref{fig:index-coding}-(b) depicts an instance of the index coding problem. With the assist of coding,   broadcasting only three coded data chunks  $d_1 + d_2$, $d_2 + d_3$, and $d_4$ (over $GF 2)$) can satisfy all clients.

\begin{figure}
	\begin{center}
		\includegraphics[width=.6\textwidth]{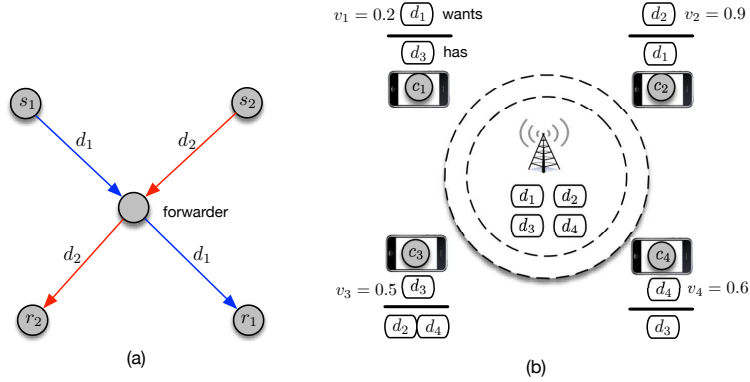}
		\caption{(a) Senders $s_1$ and $s_2$ transmit data chunks $d_1$ and $d_2$ to receivers $r_1$ and $r_2$, respectively, through a forwarder; (b) A server has a set $\mathbf{D}=\{d_1, d_2, d_3, d_4\}$ of data chunks. Client $c_i$ \textit{wants} data chunk $w_i=d_i$ with its valuation $v_i$ (as shown in the figure), and \textit{has} set $H_i \subseteq \mathbf{D}$ (as shown in the figure) as its side information.}		
		\label{fig:index-coding}
	\end{center}
\end{figure}


While transmitting an (uncoded or coded) data chunk can incur a significant transmission cost (like the energy consumption), the server transmits an (uncoded or coded) data chunk only when it is important enough to the clients. Precisely, the server transmits an (uncoded or coded) data chunk only when the overall valuation of the data chunks recovered by that transmitted data chunk can justify the transmission cost. In this context, a client has to evaluate the data chunks it wants. For example, while a client is watching an online video, it prefetches some subsequent data chunks and stores them in a prefetch queue for smoothly  playing the video. The lower the size of the prefetch queue  is, the higher the valuation of the next data chunk is. Once  the prefetch queue is empty, the client would announce a high enough valuation to justify the transmission cost. In contrast, if the size of the prefetch queue is larger, the client would announce a lower valuation and seek an opportunity that the server would transmit a coded data chunk (that can recover the data chunks the client wants) because the total valuation of the recovered data chunks can justify the transmission cost. 
	
Thus, unlike the original index coding problem (where the server has to satisfy all clients), this paper investigates the scenario where the server's transmissions strike a balance between the valuation of the data chunks and the cost of the transmissions. Instead of minimizing the number of transmissions, the first goal of this paper is to develop a coding scheme for maximizing the \textit{social welfare},  i.e., the difference between the  valuation of the data chunks that can be recovered by a client and the cost incurred by the transmissions from the server. 

To maximize the social welfare, the server needs to know each client's required data chunks,  each client's valuations of the data chunks it wants, and each client's  side information. Motivated by Fig.~\ref{fig:index-coding}-(a), the server forwards the data chunks according to the destination IP addresses; hence, it knows each client's required data chunks. However, the server does not know each client's valuations of the data chunks it wants because the data chunks are evaluated by the client itself. Moreover, the server also does not know each client's  side information because it does not fully know the client's surrounding environment like  Fig.~\ref{fig:index-coding}-(a). Thus, the server has to ask all clients to submit the unknown information. All prior works on the index coding problem assumed that the server knows the side information, for example, because the clients honestly and periodically broadcast the their side information like \cite{index-coding-3}. However, the server cannot expect a \textit{selfish} client to reveal its true information. For example, suppose that the server's strategy is simply to calculate the social welfare according to the information submitted by all clients and then to construct a code for maximizing that social welfare.  With that strategy, a selfish client would be reluctant to evaluate the data chunks it wants and simply submits  valuations as high as possible to get a higher chance of recovering the data chunks it wants. In particular,  with the wrong information, the server cannot maximize the true social welfare. 

To address that issue, an incentive for motivating each selfish client to reveal its true information is needed. This paper uses (money) payment adjustment between the server and the clients as an incentive (or a punishment).  The pricing idea has been widely used in  network design for motivating a selfish network user to behave in a prescribed way. For example, \cite{hou2010utility} developed a payment scheme for motivating each network user to submit its true utility function.   According to \cite[Theorem~9.36 or Lemma~11.9]{nisan2007algorithmic}, a payment scheme can motivate a selfish  client to submit the true valuation of a data chunk it wants (without considering the possibility that a client can lie about its side information) if and only if 1) there exists a threshold such that the client can recover the data chunk it wants if the valuation it submits is higher than the threshold, but  cannot otherwise, and 2) the client is charged the threshold if it can recover the data chunk it wants. Next, we consider a network consisting of a server and a single client that wants data chunk $d_1$ with valuation  0.6 and has data chunks $d_2, d_3$ as its side information. Following that theorem or lemma, the server might set some thresholds (satisfying the above theorem or lemma) as follows.
\begin{itemize}
	\item If the client submits the side information $\{d_2, d_3\}$, then set the threshold to 0.5. If the client submits the valuation  more than 0.5, the server transmits  $d_1+d_2+d_3$ (over $GF(2)$); otherwise, the server transmits nothing.
	\item If the client submits the side information $\{d_2\}$, then set the threshold to 0.4. If the client submits the valuation  more than 0.4, the server transmits $d_1+d_2$ (over $GF(2)$); otherwise, the server transmits nothing.
\end{itemize} 
Then, the client would submit the true valuation but the incomplete side information of  $\{d_2\}$ so that it can recover the data chunk it wants with the minimum payment 0.4. That is,  a selfish client in our problem can lie about \textit{its valuations} or \textit{the side information it has} to recover the data chunks it wants with the minimum payment.
Thus, we need to revisit those conditions for the incentive design in our problem.  The second goal of this paper is to develop a payment scheme for motivating all selfish clients to reveal their true information.




\subsection{Contributions} 
We investigate the index coding setting in the presence of selfish clients, aiming to propose a joint coding and incentive design (called a mechanism) for 1)  motivating each selfish client to truthfully reveal its side information and the valuation of each data chunk it wants  and 2)  maximizing the social welfare. 
Our first main contribution is to provide a sufficient condition for mechanisms that can motivate each selfish client to be truthful. Our second main contribution is to develop some mechanisms. With the proposed sufficient condition, we  can establish their truthfulness. Moreover, we analyze their optimality  or  worst-case approximation ratios (potentially within some restricted sets of coding matrices) in terms of the social welfare.   

%
%
%

%

\subsection{Related works}
The index coding problem was introduced in \cite{birk1998informed} and has become a hot topic. Most related works characterized  capacity regions (e.g., \cite{arbabjolfaei2019generalized,arunachala2019optimal}) for various network settings or developed computationally efficient  coding schemes to (optimally or approximately) achieve the regions (e.g., \cite{agarwal2019linear,li2018multi}). In addition to the original index coding problem,  some variants of the index coding problem have also been investigated, such as the pliable index coding  problem (e.g., \cite{brahma2015pliable}) and the secure index coding problem (e.g., \cite{mojahedian2017perfectly}). See  \cite{arbabjolfaei2018fundamentals} for extensive surveys. All the prior works on the index coding neglected potentially selfish clients. Thus, our work introduces another variant of the index coding problem by considering selfish clients. Moreover, our work is the first one to investigate the social welfare on the index coding setting.

Many prior works on network coding considered selfish clients. Most of those works  (e.g., \cite{bhadra2006min,mohsenian2013inter,zhang2008dice}) analyzed \textit{equilibrium} in the presence of selfish clients. Few works (e.g., \cite{wu2017social,chen2010inpac}) developed incentive schemes for network-coding-enabled networks. In particular, those works  focused on incentive design for \textit{fixed} coding schemes. For example, \cite{wu2017social} and \cite{chen2010inpac} used random linear codes. In contrast, our work considers a \textit{joint} coding and incentive design problem.  

Our problem is also related to combinatorial multi-item auction design (e.g., see \cite{nisan2007algorithmic}) for motivating an auction participant to reveal the true valuation of each item. However, our problem is fundamentally different from the traditional auction design. The traditional auction design can fully manage the item allocation. If an item is/isn't allocated to an auction participant, the  auction participant does/doesn't obtain it. However, because a client in our problem  can lie about its side information, the server cannot fully manage the recovery of data chunks for a client. While the server decides not to recover a data chunk that a client wants (based on the side information revealed by the client), the client might still recover the data chunk with the hidden side information. 
	Moreover, as discussed in Section~\ref{section:introduction}, the results (like  \cite[Theorem~9.36 or Lemma~11.9]{nisan2007algorithmic}) for the traditional auction design cannot immediately apply to our problem. Thus, this paper is exploring a new problem in the intersection of coding theory and game theory.

\section{System model}\label{section:model}
\subsection{Network model}
Consider a wireless broadcast network consisting of a server and a set $c_1,  \cdots, c_n$  of $n$ wireless clients, as illustrated in Fig.~\ref{fig:index-coding}-(b). The server has a set $\mathbf{D}=\{d_1, \cdots, d_m\}$ of $m$ data chunks, where each data chunk $d_i$ represents  an element of  the Galois field $GF(q)$ of order $q$.  The server can transmit uncoded data chunks or coded data chunks (combined from data chunks in set $\mathbf{D}$) to all clients over a noiseless broadcast channel. Each client  wants a single data chunk in set $\mathbf{D}$.  	 Let $w_i \in \mathbf{D}$ be the data chunk client $c_i$ wants. Multiple clients can request the same data chunk, i.e., $w_i=w_j$ for some $i$ and~$j$.  Moreover, each client already has a subset of data chunks in set $\mathbf{D}$ as its side information.  Let $H_i \subseteq  \mathbf{D}$ be the side information client $c_i$ has.


\subsection{Coding schemes}

In this paper, we consider  \textit{scalar-linear coding schemes}, where  every transmission made by the server is a linear combination of the data chunks in set $\mathbf{D}$. Precisely, the $i$-th transmission $t_i$ made by the server can be expressed by  $t_i=\sum^m_{j=1} g_{i,j} d_j$ with \textit{coding coefficient} $g_{i,j} \in GF(q)$  Let $G_i=(g_{i,1},  \cdots, g_{i,m})$ be the \emph{coding vector} of $t_i$. Moreover,  let $G=[G_i]$ be the \textit{coding matrix} whose  $i$-th rows is the coding vector of $t_i$. For a given coding matrix $G$, let function $\eta(G)$ represent the total number of transmissions made by the server.

After receiving the transmissions $t_1, \cdots, t_{\eta(G)}$ from the server, client $c_i$ can recover data chunk $w_i$ it wants if and only if there exists a (decoding) function that maps received data chunks $t_1, \cdots, t_{\eta(G)}$ and its side information $H_i$ to data chunk $w_i$.
Note that the server does not need to satisfy all clients in our setting.  For a given coding matrix $G$, let indicator function $\mathbf{1}_i(H_i, G)$ indicate if client $c_i$ can recover data chunk $w_i$ with side information $H_i$, where   $\mathbf{1}_i(H_i, G)=1$ if it can; $\mathbf{1}_i(H_i, G)=0$ if it cannot.

Each client $c_i$ has a  valuation $v_i \geq 0$ representing the importance of data chunk $w_i$ to it. Suppose that each transmission (from the server) incurs a transmission cost of one unit.  The transmission cost can reflect, for example, the power consumption. To capture the tradeoff between the importance of the data chunks and the power consumption, we define a \emph{social welfare} by 
\begin{align}
\sum_{i=1}^{n} v_i \cdot \mathbf{1}_i(H_i, G) - \eta(G), \label{eq:social-welfare}
\end{align} 
where the first term $v_i \cdot \mathbf{1}_i(H_i, G)$ expresses the valuation of data chunk $w_i$ that can be recovered by client $c_i$ and the second term $\eta(G)$ expresses the cost of the total  $\eta(G)$ transmissions made by the server. For example,  the social welfare of  transmitting $d_1+d_2$, $d_2+d_3$, and $d_4$ (the solution to the index coding problem) in Fig.~\ref{fig:index-coding}-(b)   is $0.2+0.9+0.5+0.6-3=-0.8$.  In contrast, the social welfare of transmitting  $d_3+d_4$ is $0.5+0.6-1=0.1$ (where only clients $c_3$ and $c_4$ can recover the data chunks they want). Thus, transmitting   $d_3+d_4$ is more valuable than transmitting  $d_1+d_2$, $d_2+d_3$, and $d_4$ from the global view. In this paper, we aim to develop a \textit{coding scheme} that identifies a coding matrix $G$ for maximizing the social welfare.

\subsection{Incentive schemes}
To maximize the social welfare, the server has to know  data chunk $w_i$, valuation $v_i$ and side information $H_i$ about each client $c_i$.  As discussed in Section~\ref{section:introduction},  we suppose that  the server knows $w_i$ for all $i$ but asks each client $c_i$ to submit its valuation and the \textit{indices} (but not the content) of the data chunks in its side information\footnote{Submitting the information would incur a slight transmission cost. The cost can be reflected in valuation $v_i$  while client $c_i$ evaluates data chunk $w_i$. Thus, the cost is not included in the social welfare in Eq.~(\ref{eq:social-welfare}).}. Let $\hat{v}_i>0$ and $\hat{H}_i \subseteq \mathbf{D}$ be the valuation and  the side information\footnote{If a client submits an index  that is out of the indices  of the data chunks in set $\mathbf{D}$ as its side information, then the server  neglects it (because the server only manages the data chunks in set $\mathbf{D}$). Thus, we  assume $\hat{H}_i \subseteq \mathbf{D}$ without loss of generality.}, respectively, revealed by client $c_i$. Each client $c_i$ can tell a lie, so that  $\hat{v}_i$ and $\hat{H}_i$ (obtained by the server) can be different from the true information $v_i$ and $H_i$ (owned by client $c_i$). Thus, in this paper, we also aim to develop an incentive scheme for motivating each  client $c_i$ to tell the truth so that $\hat{v}_i=v_i$ and $\hat{H}_i=H_i$ for all $i$. Let  \mbox{$\hat{\mathbf{V}}=\{\hat{v}_1, \cdots, \hat{v}_n\}$} and \mbox{$\hat{\mathbf{H}}=\{\hat{H}_1, \cdots, \hat{H}_n\}$} be the  sets of all corresponding elements.  Moreover, let $\hat{\mathbf{V}}_{-i}=\hat{\mathbf{V}}-\{\hat{v}_i\}$ and $\hat{\mathbf{H}}_{-i}=\hat{\mathbf{H}} - \{\hat{H}_i\}$ be the set of all corresponding elements except the one for client $c_i$.

In this paper, we consider money transfers between the server and the clients as an incentive. Each client $c_i$ has to pay the server for data chunk $w_i$ if the client can recover it.  In this context,  valuation $v_i$ of data chunk $w_i$ implies the maximum amount of money client $c_i$ is willing to pay to obtain  it.
Let $p_i \geq 0$ be the payment of client $c_i$ charged by the server.  A scheme determining payment $p_i$ for each client $c_i$ is  referred to as a \textit{payment scheme}. In general, a payment scheme depends on  valuation set $\hat{\mathbf{V}}$,   side information set $\hat{\mathbf{H}}$, and  coding matrix $G$ (which determines  indicator $\mathbf{1}_i(H_i, G)$ for each client $c_i$). The design of payment schemes and that of coding schemes depend on each other.  Thus, we define a \textit{mechanism}~$\pi$ by a joint coding and payment scheme.  

The mechanism used by the server is given to all clients.  For a given mechanism~$\pi$, we define a \textit{utility} for client $c_i$ by $u_i(\hat{\mathbf{V}}, \hat{\mathbf{H}},\pi)= (v_i - p_i)\cdot\mathbf{1}_i(H_i, G)$, which is the difference between the valuation of the data chunk and the money charged by the server if it can recover the data chunk it wants, but is zero otherwise. Though the utility of a client can be computed only when all clients' information is given (because the mechanism~$\pi$ needs all client's information to compute price $p_i$ and coding matrix $G$), we consider non-cooperative clients where a client has no information about  other clients and does not cooperate with other clients. In this context, we aim to develop a mechanism $\pi$ for guaranteeing that a client can maximize its utility when  submitting its true information, for any given information submitted by all other clients. 
 Mathematically, for every $\hat{\mathbf{V}}_{-i}$ and $\hat{\mathbf{H}}_{-i}$, the mechanism $\pi$ satisfies  
	\begin{align}
	u_i\bigl(\{v_i, \hat{\mathbf{V}}_{-i}\}, \{H_i, \hat{\mathbf{H}}_{-i}\}, \pi\bigr) \geq u_i\bigl(\{\hat{v}_i, \hat{\mathbf{V}}_{-i}\}, \{\hat{H}_i, \hat{\mathbf{H}}_{-i}\}, \pi\bigr), \label{eq:truthful}
\end{align} 
for all possible $\hat{v}_i$ and $\hat{H}_i$. In game theory, we call the mechanism  a  \textit{dominant strategy} \cite{mechanism-design}. A dominant strategy equilibrium is always a Nash equilibrium.


\subsection{Problem formulation}
A mechanism  satisfying Eq.~(\ref{eq:truthful}) is referred to as a \textit{truthful} mechanism. Moreover, if a truthful mechanism yields the social welfare equal to the maximum value of Eq.~(\ref{eq:social-welfare}) obtained by a coding scheme that knows the true information $\mathbf{V}$ and $\mathbf{H}$, it is referred to as an \textit{optimal} truthful mechanism.  We aim to develop an optimal truthful mechanism such that the social welfare and the utilities of all clients are simultaneously optimized.  Our problem involves both the global and  local optimization problems. 

Note that  local utility  $u_i(\hat{\mathbf{V}}, \hat{\mathbf{H}}, \pi)$ involves more than one type of private information (i.e., valuation  and  side information).  The traditional incentive design for a single type of private information might be insufficient to motivate a client in our problem to reveal the true information of both types.  Thus,  Section~\ref{section:condition} characterizes truthful mechanisms for our problem. With the results in Section~\ref{section:condition},  we will develop optimal or approximate truthful mechanisms for various  scenarios of our problem.

\section{Characterizing truthful mechanisms} \label{section:condition}
This section provides a sufficient condition of truthful mechanisms for our problem. To that end, we introduce a type of coding schemes as follows, where we  use  indicator function $\mathbf{1}_i(\hat{H}_i, G)$ to indicate if client $c_i$ can recover data chunk $w_i$ with side information $\hat{H}_i$ and coding matrix $G$.

\begin{define}
A coding scheme is a \textbf{threshold-type} coding scheme if, for every valuation set $\hat{\mathbf{V}}_{-i}$ and side information set $\hat{\mathbf{H}}$, there exists a  threshold $\bar{v}_i(\hat{\mathbf{V}}_{-i},\hat{\mathbf{H}})$ such that 
\begin{itemize}
	\item when $\hat{v}_i > \bar{v}_i(\hat{\mathbf{V}}_{-i},\hat{\mathbf{H}})$, it constructs a coding matrix $G$ such that  $\mathbf{1}_i(\hat{H}_i, G)=1$;
	\item when $\hat{v}_i < \bar{v}_i(\hat{\mathbf{V}}_{-i},\hat{\mathbf{H}})$, it constructs a coding matrix $G$ such that $\mathbf{1}_i(\hat{H}_i, G)=0$,
\end{itemize}
for all $i$.
\end{define}

Note that threshold $\bar{v}_i(\hat{\mathbf{V}}_{-i},\hat{\mathbf{H}})$ for client $c_i$ is independent of valuation $\hat{v}_i$ submitted by client $c_i$, but is dependent on side information $\hat{H}_i$ submitted by client $c_i$.  The next theorem provides a  sufficient condition of  truthful mechanisms for our problem.

\begin{theorem} \label{theorem:truthfulness}
Suppose that $\hat{H}_i \subseteq H_i$ for all $i$. A mechanism is truthful if the following four conditions hold:
\begin{enumerate}
	\item The coding scheme is a threshold-type coding scheme.
	\item  The payment scheme determines payment $p_i=\bar{v}_i(\hat{\mathbf{V}}_{-i},\hat{\mathbf{H}})$ for  client $c_i$ if $\mathbf{1}_i(\hat{H}_i, G)=1$, or $p_i=0$  if  $\mathbf{1}_i(\hat{H}_i, G)=0$, for all coding matrices $G$ constructed by the coding mechanism.
	\item   For every $\hat{\mathbf{V}}_{-i}$ and $\hat{\mathbf{H}}_{-i}$, $\bar{v}_i(\hat{\mathbf{V}}_{-i},\{H_i, \hat{\mathbf{H}}_{-i}\}) \leq \bar{v}_i(\hat{\mathbf{V}}_{-i},\{\hat{H}_i, \hat{\mathbf{H}}_{-i}\})$ for all $\hat{H}_i \subseteq H_i$.
	\item For every $\hat{\mathbf{V}}_{-i}$ and  $\hat{\mathbf{H}}_{-i}$, if the coding scheme can construct  a coding matrix $G$ such that $\mathbf{1}_i(\hat{H}_i, G)=0$ but $\mathbf{1}_i(H_i, G)=1$ for some $\hat{v}_i$ and  $\hat{H}_i$, then $\bar{v}_i(\hat{\mathbf{V}}_{-i},\{H_i, \hat{\mathbf{H}}_{-i}\})=0$.		
\end{enumerate}
\end{theorem}
\begin{proof}
See Appendix \ref{appendix:theorem:truthfulness}.
\end{proof}

The first two conditions in the above theorem claim that a client cannot affect the payment by  lying about the valuation of the data chunk it wants, because its payment depends on  the valuations  submitted by other clients. The third condition claims that  a client can minimize its payment when  submitting its complete side information. The fourth condition further  considers the case when the server cannot fully manage the recovery of the data chunk a client wants, i.e., the server decides not to recover $w_i$ (based on $\hat{v}_{i}$ and $\hat{H}_i$) but client $c_i$ can still recover it by the hidden side information $H_i-\hat{H}_i$. For that case, the fourth condition  claims that the client can obtain the data chunk it wants for free when  submitting its true side information. The theorem will be used later to establish the truthfulness of the proposed mechanisms. We remark that the theorem generalizes the  sufficient condition in \cite[Theorem~9.36 or Lemma~11.9]{nisan2007algorithmic} to the case when a client can lie about the side information it has.

Note that Theorem~\ref{theorem:truthfulness} assumes $\hat{H}_i \subseteq H_i$ for all $i$. To avoid the case when $\hat{H}_i \not\subseteq H_i$ (i.e., client $c_i$ announces  data chunks it does not really have as its side information), the server can use a hash function for \textit{validating} if a client really has the content of a data chunk. The hash function  takes the index and the first few bits of a data chunk as an input, performs some operations on it, and returns a value as the output.  For example, the hash function can be SHA-256, which has been widely used in the Bitcoin protocol for \textit{validating} a transaction \cite{nakamoto2019bitcoin} An important property of the hash function is that  a slightly different input can produce a  completely different output value. That is, the hash output value for data chunk $d_i$ is different for all $i$.  The server can compute the mapping between the data chunks in set $\mathbf{D}$ and their respective hash output values in advance. Then, through a hash output value, the server can know the associated input data chunk. 

Thus, the server asks each client to submit the hash output values,  instead of submitting the indices of its side information.  If a client wants to lie to the server that it has a data chunk in set $\mathbf{D}-H_i$, it could only take the index and some  randomly guessed bits (as it does not really have the content of that data chunk) as an input of the hash function. However, because the random bits are unlikely to be exactly the same as those bits in a data chunk in set $\mathbf{D}-H_i$, the hash output value is unlikely to be that produced by a data chunk in set $\mathbf{D}-H_i$. 
Through the hash output values, the server can identify the side information that a client announces but it does not really have. Thus,  with the assist of the hash function techniques, a client does not want to lie about a data chunk it does not really have as its side information. The rest of this paper assumes $\hat{H}_i \subseteq H_i$ for all $i$ and aims to motivate each client to submit its \textit{complete} side information so that  $\hat{H}_i = H_i$ for all~$i$.

\section{VCG-based mechanism design} 
\label{section:vcg-coding}

This section proposes an optimal truthful mechanism leveraging the celebrated Vickrey-Clarke-Groves (VCG) approach \cite{nisan2007algorithmic}. Note that the original VCG mechanism  provides an auction participant with an incentive to reveal only the true valuation of each item. However, our Theorem~\ref{theorem:truthful-coding-pricing} will show that the proposed VCG-based mechanism can motivate each selfish client to reveal not only the true valuation of the data chunk it wants but also its complete side information.

Our VCG-based mechanism uses the following  function  
\begin{align}
w(\hat{\mathbf{V}}, \hat{\mathbf{H}},G)= \sum_{i=1}^{n} \hat{v}_i \cdot \mathbf{1}_i(\hat{H}_i, G) -\eta(G). \label{eq:w}
\end{align}
The function $w(\hat{\mathbf{V}}, \hat{\mathbf{H}},G)$ is  the social welfare in Eq.~(\ref{eq:social-welfare}) computed by the information $\hat{\mathbf{V}}$ and $\hat{\mathbf{H}}$ obtained by the server.  Then, we propose our VCG-based mechanism as follows, including a VCG-based coding scheme and a VCG-based payment scheme, for computing a coding matrix and payments when obtaining valuation set $\hat{\mathbf{V}}$ and side information set $\hat{\mathbf{H}}$ from the clients. 

\textbf{VCG-based coding scheme:} Identify a coding matrix $G^*$ for maximizing  function $w(\hat{\mathbf{V}}, \hat{\mathbf{H}}, G)$:
\begin{align} \label{eq:vcg-coding}
	G^* \in \arg\max_{G \in \mathbf{G}} w(\hat{\mathbf{V}},\hat{\mathbf{H}}, G),
\end{align}
where $\mathbf{G}$ is a set of the coding matrices that can be selected.
If there is a tie in Eq.~(\ref{eq:vcg-coding}), it is broken arbitrarily.

\textbf{VCG-based payment scheme:} If $\mathbf{1}_i(\hat{H}_i, G^*)=0$ for coding matrix $G^*$ computed by Eq.~(\ref{eq:vcg-coding}), then  client $c_i$ is charged $p_i=0$; otherwise, it is charged 
\begin{align} \label{eq:vcg-pricing}
	p_i = \max_{G \in \mathbf{G}} w(\{0, \hat{\mathbf{V}}_{-i}\}, \hat{\mathbf{H}}, G) - w(\{0, \hat{\mathbf{V}}_{-i}\}, \hat{\mathbf{H}}, G^*),
\end{align}
where valuation set $\{0, \hat{\mathbf{V}}_{-i}\}$ is valuation set $\hat{\mathbf{V}}$ with valuation $v_i$ being substituted by zero.  Because of the optimality of the first term in Eq.~(\ref{eq:vcg-pricing}), payment $p_i$ is non-negative. The idea underlying Eq.~(\ref{eq:vcg-pricing}) is to calculate  threshold $\bar{v}_i(\hat{\mathbf{V}}_{-i},\hat{\mathbf{H}})$. Suppose that client $c_i$ submits valuation $\tilde{v}_i > \max_{G \in \mathbf{G}} w(\{0, \hat{\mathbf{V}}_{-i}\}, \hat{\mathbf{H}}, G) - w(\{0, \hat{\mathbf{V}}_{-i}\}, \hat{\mathbf{H}}, G^*)$. Then, we can obtain
		\begin{align*}
		w(\{\tilde{v}_i, \hat{\mathbf{V}}_{-i}\}, \hat{\mathbf{H}}, G^*)\mathop{=}^{(a)}&w(\{0, \hat{\mathbf{V}}_{-i}\}, \hat{\mathbf{H}}, G^*)+\tilde{v}_i\\
		\mathop{>}^{(b)} & \max_{G \in \mathbf{G}} w(\{0, \hat{\mathbf{V}}_{-i}\}, \hat{\mathbf{H}}, G), 
	\end{align*}	
where (a) is from Eq.~(\ref{eq:w}); (b) is from the assumption of $\tilde{v}_i$. In the above inequality,  the term $\max_{G \in \mathbf{G}} w(\{0, \hat{\mathbf{V}}_{-i}\}, \hat{\mathbf{H}}, G)$ is the maximum function value among all possible coding matrices $G$ such that $\mathbf{1}_i(\hat{H}_i, G)=0$. Moreover, we want to emphasize that Eq.~(\ref{eq:vcg-pricing}) calculates payment $p_i$ for coding matrix $G^*$ such that $\mathbf{1}_i(\hat{H}_i, G^*)=1$. Thus, the inequality implies that  the VCG-based coding scheme constructs a coding matrix $G$ such that $\mathbf{1}_i(\hat{H}_i, G)=1$ when client $c_i$ submits a valuation greater than the value computed by Eq.~(\ref{eq:vcg-pricing}), implying the first and second conditions in Theorem~\ref{theorem:truthfulness}.

The idea behind why the VCG-based mechanism can satisfy the third in Theorem~\ref{theorem:truthfulness} is that for a fixed valuation $\tilde{v}_i$, if client $c_i$ reveals more data chunks as its side information $\hat{H}_i$, the VCG-based coding scheme is more likely to construct a coding matrix $G^*$ such that $\mathbf{1}_i(\hat{H}_i, G^*)=1$. That is, the threshold decreases with the size $|\hat{H}_i|$, implying the third condition in Theorem~\ref{theorem:truthfulness}.  Moreover, if the VCG-based coding scheme can construct a coding matrix $G^*$ such that $\mathbf{1}_i(\hat{H}_i,G^*)=0$ but $\mathbf{1}_i(H_i,G^*)=1$, then when client $c_i$ submits zero valuation $\tilde{v}_i=0$ and its complete side information $\hat{H}_i=H_i$, the  coding matrix $G^*$ is also a maximizer in Eq.~(\ref{eq:vcg-coding}) (and $\mathbf{1}_i(H_i, G^*)=1$). That is, the threshold is zero, implying the fourth condition in Theorem~\ref{theorem:truthfulness}.

The next theorem  establishes the truthfulness and the optimality of the proposed VCG-based mechanism.

\begin{theorem} \label{theorem:truthful-coding-pricing}  
The VCG-based mechanism is an optimal  truthful mechanism.
\end{theorem}
\begin{proof}
Appendix \ref{appendix:truthful-coding-pricing} confirms that the VCG-based mechanism is truthful (by Theorem~\ref{theorem:truthfulness}). Then, all clients submit the true valuations of the data chunks they want and their complete side information. Moreover, by Eq.~(\ref{eq:vcg-coding}), the VCG-coding scheme  maximizes the social welfare. Thus, the VCG-based mechanism is an optimal truthful mechanism. 
\end{proof}

The proposed VCG-based mechanism involves the combinatorial optimization problems in both Eqs.~(\ref{eq:vcg-coding}) and~(\ref{eq:vcg-pricing}). Appendix \ref{appendix:proposition:general-hard} shows (by a reduction from the original index coding problem) that the combinatorial optimization problems are NP-hard. To develop computationally efficient mechanisms, the rest of this paper focuses on \textit{sparse coding schemes}, which construct coding matrices $G$ (over $GF(2)$) such that at most two coding coefficients in each coding vector $G_i$ is nonzero. 	That is, a sparse coding scheme combines at most two data chunks based on a small field size $GF(2)$ for each transmission, resulting in a smaller packet size and a lower  encoding/decoding complexity. Some  sparse coding schemes have been also developed (e.g., \cite{chaudhry2011finding}) for approximately minimizing the number of transmissions in the original index coding problem.

This paper will  consider two different scenarios separately: the \textit{multiple unicast scenario} and the \textit{multiple multicast scenario}. While in the  multiple unicast scenario each client  wants a different data chunk (i.e., $n=m$), in the  multiple multicast scenario many clients can request the same data chunk (i.e., $n \geq m$). Moreover, this paper will also  consider two different decoding schemes separately: the \textit{instant decoding scheme} and the \textit{general decoding scheme}. While an instant decoding scheme can combine each \textit{individual} transmission (from the server) with its side information but cannot combine multiple transmissions, a general decoding scheme can combine more than one transmission with its side information. For example, in Fig.~\ref{fig:index-coding},  client $c_2$ can instantly decode  data chunk $d_2$ by   $d_1+d_2$; however,  it cannot instantly decode  data chunk $d_1$  by    $d_1+d_2$ or  $d_2+d_3$ separately (but it can decode $d_1$ by combining both  $d_1+d_2$ and $d_2+d_3$). The instant decoding scheme requires a client to store at most one transmission  for recovering the  data chunk it wants, resulting in a smaller buffer size and a lower decoding complexity. This paper has those partitions such that each section serves a certain set of clients.

Section~\ref{section:muliple-unicast} develops computationally efficient mechanisms for the multiple unicast scenario. While Section~\ref{subsection:multiple-unicast-instant} proposes an algorithm  optimally solving Eqs.~(\ref{eq:vcg-coding}) and~(\ref{eq:vcg-pricing}) in polynomial time  when the set $\mathbf{G}$ is restricted to sparse and instantly decodable coding matrices, Section~\ref{subsection:approximation-alg1} establishes that the combinatorial optimization problems in Eqs.~(\ref{eq:vcg-coding}) and (\ref{eq:vcg-pricing}) are still NP-hard  when the set $\mathbf{G}$ is restricted to only sparse  coding matrices. To cope with the NP-hardness, Sections~\ref{subsection:approximation-alg1} and~\ref{subsection:approximation-alg2} develop two approximate truthful mechanisms. Moreover, Section~\ref{subsection:complexity} analyzes the computational complexities of the proposed polynomial-time coding schemes. Subsequently, Section~\ref{section:multiple-multicast} shows that the combinatorial optimization problems in Eqs.~(\ref{eq:vcg-coding}) and~(\ref{eq:vcg-pricing})  for the multiple multicast scenario are not only NP-hard but also NP-hard to approximate even using those simple sparse coding schemes.  Table~\ref{table:alg-hard} summarizes our main results where the symbol $C$ will be defined soon.

\begin{table*}[!t]
	\footnotesize
	\caption{Summary of the proposed mechanisms}
	\begin{center}
		\begin{tabular}{| >{\raggedright\arraybackslash}m{2cm}| >{\raggedright\arraybackslash}m{2.5cm}| >{\raggedright\arraybackslash}m{3cm}| >{\raggedright\arraybackslash}m{2.2cm}| >{\raggedright\arraybackslash}m{2.5cm}| >{\raggedright\arraybackslash}m{2cm}|}
			\hline
			\textbf{proposed coding schemes} & \textbf{applicable scenarios} & \textbf{restrictions on  set $\mathbf{G}$} &\textbf{optimality results in the restricted  set $\mathbf{G}$}& \textbf{complexity results} & \textbf{associated payment scheme}\\
			\hline
			VCG-based coding scheme& multiple unicast, multiple multicast & general set & optimal &NP-hard, NP-hard to approximate in multiple multicast & VCG-based payment scheme\\
			\hline
			VCG-based coding scheme along with Alg.~\ref{alg:poly-time} & multiple unicast & sparse and instantly decodable coding matrices &optimal & $O(n^3)$ & VCG-based payment scheme along with Alg.~\ref{alg:poly-time}  \\
			\hline			Alg.~\ref{alg:greedy-vcg-coding} & multiple unicast  & sparse coding matrices  & $\max|C|$-approximate &  $O(n^4)$ & Alg.~\ref{alg:greedy-vcg-pricing}\\
			\hline
			revised Alg.~\ref{alg:greedy-vcg-coding} with Alg.~\ref{alg:max-sqrt-weight} & multiple unicast& sparse coding matrices & $\sqrt{n}$-approximate &  $O(n^5)$ & revised Alg.~\ref{alg:greedy-vcg-pricing}  with Alg.~\ref{alg:max-sqrt-weight} \\
			\hline					
		\end{tabular}
		\label{table:alg-hard}
	\end{center}
\end{table*}

%
%
%

\section{Mechanism design for the multiple unicast scenario} \label{section:muliple-unicast}

This section develops computationally efficient truthful mechanisms  for the multiple unicast scenario by  proposing polynomial-time algorithms for  (optimally or approximately) solving Eq.~(\ref{eq:vcg-coding}) within the set $\mathbf{G}$ of sparse coding matrices. We remark that an approximate solution to Eqs.~(\ref{eq:vcg-coding}) and~(\ref{eq:vcg-pricing}) is no longer a truthful mechanism (see Example~\ref{ex:counter} later). Thus, we devise alternative payment schemes to substitute the  previously proposed VCG-based payment scheme for guaranteeing the truthfulness (see Sections~\ref{subsection:approximation-alg1} and~\ref{subsection:approximation-alg2} later).

To solve Eq.~(\ref{eq:vcg-coding}),  we introduce  a  \textit{weighted dependency graph} constructed as follows:  given valuation set $\hat{\mathbf{V}}$ and side information set $\hat{\mathbf{H}}$, 
\begin{itemize}
	\item for each client $c_i$, construct  a  vertex $\lambda_i$;
	\item for any two clients $c_i$ and $c_j$ such that $w_i \in \hat{H}_j$, construct a directed arc $(\lambda_i, \lambda_j)$;
	\item associate each arc $(\lambda_i, \lambda_j)$ with an arc weight  $\gamma_{(\lambda_i, \lambda_j)}=\hat{v}_i$. 
\end{itemize}
The weighted dependency graph generalizes the dependency graph in \cite{birk1998informed} to a weighted version. 
We denote the weighted dependency graph by $\mathcal{G}(\mathbf{\Lambda}, \mathbf{A}, \mathbf{\Gamma})$, where $\mathbf{\Lambda}$ is the vertex set, $\mathbf{A}$ is the arc set, and $\mathbf{\Gamma}$ is the arc weight set.   Fig.~\ref{fig:dep-graph} illustrates the weighted dependency graph for the instance in  Fig.~\ref{fig:index-coding}. 

\begin{figure}[t]
	\begin{center}
		\includegraphics[width=.3\textwidth]{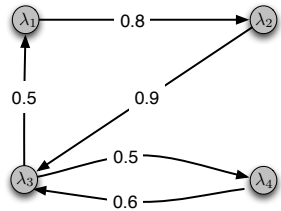}
		\caption{Weighted dependency graph for the instance  in Fig. \ref{fig:index-coding} when $\hat{v}_1=0.8$, $\hat{v}_2=0.9$, $\hat{v}_3=0.5$, $\hat{v}_4=0.6$, and $\hat{H}_i=H_i$ for all $i$.}
		\label{fig:dep-graph}
	\end{center}
\end{figure}

We make two  observations about weighted dependency graphs:
\begin{itemize}
		\item \textit{for the general decoding scheme}, the server can satisfy all clients in a cycle $C$ in graph $\mathcal{G}(\mathbf{\Lambda}, \mathbf{A}, \mathbf{\Gamma})$ with $|C|-1$  sparse coded data chunks;
	
	\item \textit{for the instant decoding scheme}, 	the server can satisfy all clients in a cycle $C$ with $|C|=2$ in graph $\mathcal{G}(\mathbf{\Lambda}, \mathbf{A}, \mathbf{\Gamma})$ with $|C|-1=1$ sparse coded data chunk.
\end{itemize} 
For example, with the general decoding scheme, clients $c_1$, $c_2$, and $c_3$ in cycle $(\lambda_1, \lambda_2, \lambda_3)$ of Fig.~\ref{fig:dep-graph}  can recover the data chunks they want with  $d_1 +d_2$ and $d_2 +d_3$. In contrast, with the instant decoding scheme, clients $c_1$, $c_2$, and $c_3$  cannot recover the data chunks they want with  any two transmissions among $d_1 +d_2$, $d_2 +d_3$, or $d_3+d_1$. However, clients $c_3$ and $c_4$ in cycle $(\lambda_3, \lambda_4)$ of Fig.~\ref{fig:dep-graph}  can instantly decode the data chunks they want with   $d_3 +d_4$.

We say that a coding scheme \textit{encodes along cycle $C$} in weighted dependency graph $\mathcal{G}(\mathbf{\Lambda}, \mathbf{A}, \mathbf{\Gamma})$ if it constructs  $|C|-1$ sparse coded data chunks for satisfying all clients in the cycle. Note that, for the instant decoding scheme, a coding scheme can  encode along cycle $C$ with $|C|=2$ only, according to the above observations. While encoding along a set of (vertex) disjoint cycles can satisfy all clients in those cycles with fewer transmissions than the number of the satisfied clients, all other sparse codes with no cycle being involved cannot (see \cite{chaudhry2011finding} for details). Those  transmissions with no cycle being involved can be substituted by uncoded data chunks without changing the function value in Eq.~(\ref{eq:w}). Thus, we can focus on sparse coding schemes that encodes along disjoint cycles and additionally transmits uncoded data chunk $w_i$ if vertex $\lambda_i$ is not in those cycles (i.e., client $c_i$ cannot recover data chunk $w_i$ with the coded data chunks along the cycles) but $\hat{v}_i \geq 1$.

We aim to identify a sparse coding matrix $G$ including the coding vectors  along a set $\mathbf{C}$ of disjoint cycles and  those of uncoded data chunk $w_i$ if $\lambda_i$ in not in those cycles but $\hat{v}_i \geq 1$, for maximizing 
\begin{align}
w(\hat{\mathbf{V}}, \hat{\mathbf{H}}, G)=\sum_{C \in \mathbf{C}}\underbrace{\left( \sum_{\lambda_i \in C} \hat{v}_i-(|C|-1)\right)}_{(a)}+\underbrace{\sum_{\substack{\lambda_i \notin C, \forall C\in \mathbf{C} \\ \hat{v}_i\geq 1}} (\hat{v}_i-1)}_{(b)}, \label{eq:w-cycle1}
\end{align}
where (a) is because encoding along cycle $C$ satisfies all clients in the cycle with $|C|-1$ transmissions; (b) considers uncoded data chunks for those clients that submit valuations of no less than one but cannot recover the data chunks they want with the coded data chunks along the cycles.  Then, we use the notation $[x]^+_1=\min\{x,1\}$ to represent the truncation of $x$ toward one; in particular, we can re-write Eq.~(\ref{eq:w-cycle1}) in terms of  truncated valuations as follow:
\begin{align}
w(\hat{\mathbf{V}}, \hat{\mathbf{H}}, G)=\sum_{C \in \mathbf{C}}\left( \sum_{\lambda_i \in C} [\hat{v}_i]^+_1-(|C|-1)\right)+\sum_{\substack{\lambda_i \notin C, \forall C\in \mathbf{C} \\ \hat{v}_i\geq 1}} ([\hat{v}_i]^+_1-1)+\underbrace{\sum_{\hat{v}_i \geq 1} (\hat{v}_i-1)}_{(a)}, \label{eq:w-cycle2}
\end{align}
where (a) adds back the deducted value (caused by the truncation). Because the value of the term  $\sum_{\substack{\lambda_i \notin C, \forall C\in \mathbf{C} \\ \hat{v}_i\geq 1}} ([\hat{v}_i]^+_1-1)$ is zero and the value of the term $\sum_{\hat{v}_i \geq 1} (\hat{v}_i-1)$  is constant, it suffices to  maximize  $\sum_{C \in \mathbf{C}}(\sum_{\lambda_i \in C} [\hat{v}_i]^+_1-(|C|-1))$. 

To that end,  we associate  each cycle $C$  in weighted dependency graph $\mathcal{G}(\mathbf{\Lambda}, \mathbf{A}, \mathbf{\Gamma})$ with a cycle weight $\gamma(C)$ defined by
\begin{align}
\gamma(C)=\sum_{a \in C} [\gamma_a]^+_1- (|C|-1),  \label{eq:cycle-weight}
\end{align}
which implies the difference between the total truncated valuation submitted by the clients  in cycle $C$ and the  cost of encoding along cycle $C$. Note that, for the instant decoding scheme, we  assign cycle weight $\gamma(C)$ to cycles $C$ with $|C|=2$ only. Then,  we can turn our attention to a maximum weight cycle packing problem: identifying a set $\mathbf{C}$ of disjoint cycles for maximizing the total cycle weight $\sum_{C \in \mathbf{C}} \gamma(C)$ in graph $\mathcal{G}(\mathbf{\Lambda}, \mathbf{A}, \mathbf{\Gamma})$. 

Section~\ref{subsection:multiple-unicast-instant} optimally solves our maximum weight cycle packing problem for the instant decoding scheme. For the general decoding scheme, Sections~\ref{subsection:approximation-alg1} and~\ref{subsection:approximation-alg2} propose two approximate solutions to our maximum weight cycle packing problem and their respective payment schemes as the incentives.

\subsection{The instant decoding scheme} \label{subsection:multiple-unicast-instant}
This section develops Alg.~\ref{alg:poly-time} for optimally solving Eqs.~(\ref{eq:vcg-coding})  and (\ref{eq:vcg-pricing}) when set $\mathbf{G}$ is restricted to those sparse and instantly decodable coding matrices. Note that Theorem~\ref{theorem:truthful-coding-pricing} holds for any set $\mathbf{G}$. Thus, the VCG-based mechanism along with Alg.~\ref{alg:poly-time} is an optimal truthful mechanism in the set $\mathbf{G}$ of sparse and instantly decodable coding matrices.

Given valuation set $\hat{\mathbf{V}}$ and side information set $\hat{\mathbf{H}}$, Alg.~\ref{alg:poly-time} aims to construct a sparse and instantly decodable coding matrix $G$   for maximizing function $w(\hat{\mathbf{V}}, \hat{\mathbf{H}}, G)$ in Eq.~(\ref{eq:w-cycle2}) in polynomial time. To that end, Alg.~\ref{alg:poly-time} constructs weighted dependency graph $\mathcal{G}(\mathbf{\Lambda}, \mathbf{A}, \mathbf{\Gamma})$ in Line~\ref{alg:poly-time:g1}, aiming to identify a set $\mathbf{C}$ of disjoint cycles $C$ with $|C|=2$ for maximizing total cycle weight $\sum_{C \in \mathbf{C}} \gamma(C)$ in set $\mathbf{C}$. To identify such a set of disjoint cycles, Alg.~\ref{alg:poly-time} constructs an undirected graph $\mathcal{G}(\tilde{\mathbf{\Lambda}},\tilde{\mathbf{E}}, \tilde{\mathbf{\Lambda}})$ in Line~\ref{alg:poly-time:g2} with the following procedure:
\begin{itemize}
	\item for each vertex $\lambda\in \mathbf{\Lambda}$, construct a vertex $\tilde{\lambda}\in \tilde{\mathbf{\Lambda}}$;
	\item for any two vertices $\lambda_i, \lambda_j \in \mathbf{\Lambda}$ such that both arcs $(\lambda_i, \lambda_j)$ and $(\lambda_j, \lambda_i)$ are in set $\mathbf{A}$, construct an edge $(\tilde{\lambda}_i, \tilde{\lambda}_j) \in \tilde{\mathbf{E}}$;
	\item associate each edge $(\tilde{\lambda}_i, \tilde{\lambda}_j) \in \tilde{\mathbf{E}}$ with an edge weight $\tilde{\gamma} \in \tilde{\mathbf{\Gamma}}$ such that $\tilde{\gamma}_{(\tilde{\lambda}_i, \tilde{\lambda}_j)}=[\gamma_{(\lambda_i,\lambda_j)}]^+_1+[\gamma_{(\lambda_j,\lambda_i)}]^+_1-1$.
\end{itemize}
With the construction, each cycle $C$ with $|C|=2$ in graph $\mathcal{G}(\mathbf{\Lambda}, \mathbf{A}, \mathbf{\Gamma})$  corresponds to  an edge in graph $\mathcal{G}(\tilde{\mathbf{\Lambda}},\tilde{\mathbf{E}}, \tilde{\mathbf{\Gamma}})$;  in particular, a set of disjoint cycles in  graph $\mathcal{G}(\mathbf{\Lambda}, \mathbf{A}, \mathbf{\Gamma})$  corresponds to  a matching in graph $\mathcal{G}(\tilde{\mathbf{\Lambda}},\tilde{\mathbf{E}}, \tilde{\mathbf{\Gamma}})$.
Moreover, a cycle weight in graph $\mathcal{G}(\mathbf{\Lambda}, \mathbf{A}, \mathbf{\Gamma})$ corresponds to the edge weight in graph $\mathcal{G}(\tilde{\mathbf{\Lambda}},\tilde{\mathbf{E}}, \tilde{\mathbf{\Gamma}})$. Thus, a set of disjoint cycles in graph $\mathcal{G}(\mathbf{\Lambda}, \mathbf{A}, \mathbf{\Gamma})$ for maximizing the total cycle weight corresponds to a maximum weight matching in graph $\mathcal{G}(\tilde{\mathbf{\Lambda}},\tilde{\mathbf{E}}, \tilde{\mathbf{\Gamma}})$. Alg.~\ref{alg:poly-time} identifies a maximum weight matching in graph $\mathcal{G}(\mathbf{\Lambda}, \mathbf{A}, \mathbf{\Gamma})$ in Line~\ref{alg:poly-time:max-matching} (by some polynomial-time algorithms like the Edmonds's algorithm \cite{galil1986efficient}). Subsequently, Alg.~\ref{alg:poly-time}  adds the coding vectors of the coded data chunks (along those cycles corresponding to the maximum weight matching) to coding matrix $G^*$ in Line~\ref{alg:poly-time:coding}. Finally, if client $c_i$ submitting valuation $\hat{v}_i \geq 1$ is not satisfied by the coding matrix constructed by the maximum weight matching,  then Alg.~\ref{alg:poly-time} adds the coding vector of  data chunk $w_i$ to coding matrix $G^*$ in Line~\ref{alg:poly-time:last}. The discussion in this paragraph  leads to the following result: Alg. \ref{alg:poly-time} can optimally solve the combinatorial optimization problems in Eqs. (\ref{eq:vcg-coding}) and (\ref{eq:vcg-pricing}) in polynomial time in the set $\mathbf{G}$ of sparse and instantly decodable coding matrices.

\begin{algorithm}[t]
	\SetCommentSty{text}
	\SetAlgoLined 
	\SetKwFunction{Union}{Union}\SetKwFunction{FindCompress}{FindCompress} \SetKwInOut{Input}{input}\SetKwInOut{Output}{output}
	
	\Input{Valuation set $\hat{\mathbf{V}}$ and  side information set $\hat{\mathbf{H}}$.}
	\Output{Sparse and instantly decodable coding matrix $G^*$ for maximizing function $w(\hat{\mathbf{V}}, \hat{\mathbf{H}}, G)$.}
	
	$G^* \leftarrow \emptyset$\;
	Construct  weighted dependency graph $\mathcal{G}(\mathbf{\Lambda}, \mathbf{A}, \mathbf{\Gamma})$\; \label{alg:poly-time:g1}
	
	Construct undirected auxiliary graph $\mathcal{G}(\tilde{\mathbf{\Lambda}},\tilde{\mathbf{E}}, \tilde{\mathbf{\Gamma}})$\;\label{alg:poly-time:g2}
	
	Find a maximum weight matching $M^*$ in $\mathcal{G}( \tilde{\mathbf{\Lambda}}, \tilde{\mathbf{E}}, \tilde{\mathbf{\Gamma}})$\; \label{alg:poly-time:max-matching}
	For each edge $(\tilde{\lambda}_i, \tilde{\lambda}_j) \in M^*$, add the coding vector of $w_i +w_j$ to coding matrix $G^*$\; \label{alg:poly-time:coding}

	For each vertex $\lambda_i \notin M^*$  but $\hat{v}_i\geq 1$, add the coding vector of data chunk $w_i$ to coding matrix $G^*$\; \label{alg:poly-time:last}
	
	\caption{Polynomial-time algorithm for solving Eqs.~(\ref{eq:vcg-coding}) and (\ref{eq:vcg-pricing}) in the multiple unicast scenario.}
	
	\label{alg:poly-time}
\end{algorithm}

\subsection{General decoding scheme: $\max |C|$-approximate truthful mechanism} \label{subsection:approximation-alg1}
Appendix~\ref{appendix:lemma:multiple-unicast-hard} shows  (by a reduction from the cycle packing problem \cite{cycle-packing}) that the combinatorial optimization problems in Eqs.~(\ref{eq:vcg-coding}) and~(\ref{eq:vcg-pricing}) in the set $\mathbf{G}$ of sparse and general decodable coding matrices are still NP-hard. Thus, this section and next section develop two algorithms (Alg.~\ref{alg:greedy-vcg-coding} and its further modification)  for approximately solving  Eq.~(\ref{eq:vcg-coding}) in the set $\mathbf{G}$ of sparse coding matrices. To that end, Alg.~\ref{alg:greedy-vcg-coding} constructs weighted dependency graph $\mathcal{G}(\mathbf{\Lambda}, \mathbf{A}, \mathbf{\Gamma})$ in Line~\ref{algorithm1:graph}, aiming to approximately solving our maximum weight cycle packing problem.

\begin{algorithm}[t]
\SetCommentSty{text}
\SetAlgoLined 
\SetKwFunction{Union}{Union}\SetKwFunction{FindCompress}{FindCompress} \SetKwInOut{Input}{input}\SetKwInOut{Output}{output}

\Input{Valuation set $\hat{\mathbf{V}}$ and side information set $\hat{\mathbf{H}}$.}
\Output{Sparse coding matrix $G$.}

$G \leftarrow \emptyset$\;
 
Construct   weighted dependency graph $\mathcal{G}(\mathbf{\Lambda}, \mathbf{A}, \mathbf{\Gamma})$\;  \label{algorithm1:graph}

Associate each arc $a \in \mathbf{A}$ with arc cost $\zeta_a=1-[\gamma_a]^+_1$\; \label{algorithm1:cost}

 
\While{there is a cycle in the present graph $\mathcal{G}(\mathbf{\Lambda}, \mathbf{A}, \mathbf{\Gamma})$ whose cycle cost is less than or equal to one \label{algorithm1:condition}}{
	Find a cycle $C$ in the present graph $\mathcal{G}(\mathbf{\Lambda}, \mathbf{A}, \mathbf{\Gamma})$ for minimizing cycle cost $\zeta(C)$\;\label{algorithm1:greedy}
%
	Add the coding vectors of the $|C|-1$ coded data chunks along   cycle $C$ to coding matrix $G$\; \label{algorithm1:add-coding}
	Remove all vertices in cycle $C$ and their incident arcs  from the  present graph $\mathcal{G}(\mathbf{\Lambda}, \mathbf{A}, \mathbf{\Gamma})$\; \label{algorihtm1:remove}
} 


\For{$i \leftarrow 1$ \KwTo $n$}{
	\If{$\lambda_i$ is not in the selected cycles but $\hat{v}_i \geq 1$}{
		Add the coding vector of  data chunk $w_i$ to coding matrix $G$\;\label{algorithm1:serve2} 
	}

}

\caption{$\max |C|$-approximate coding scheme for the multiple unicast scenario.}

\label{alg:greedy-vcg-coding}
\end{algorithm}

%


The idea underlying Alg.~\ref{alg:greedy-vcg-coding} is to iteratively identify  a maximum weight cycle in a greedy way. Note that, in general, identifying a maximum weight cycle in a graph is NP-hard \cite{schrijver2003combinatorial}. However, for our problem, we can observe that cycle weight $\gamma(C)$ of cycle $C$ in Eq.~(\ref{eq:cycle-weight}) can be written as 
\begin{align}
\gamma(C)=1- \sum_{a \in C}(1-[\gamma_a]^+_1). \label{eq:cycle-weight-2}
\end{align}
By associating each arc $a \in \mathbf{A}$ with an arc cost $\zeta_{a}=1-[\gamma_a]^+_1$, we can associate  each cycle $C$ with a cycle cost $\zeta(C)=\sum_{a \in C} \zeta_{a}$,
which is the total arc cost in  cycle $C$. Then, cycle weight $\gamma(C)$  in Eq.~(\ref{eq:cycle-weight-2}) becomes $\gamma(C)=1-\zeta(C)$. Removing the constant, a maximum weight cycle $C$ minimizes cycle cost $\zeta(C)$. Thus, Alg.~\ref{alg:greedy-vcg-coding} identifies a minimum cost cycle in Line~\ref{algorithm1:greedy}  (by some polynomial-time algorithms like the Floyd-Warshall algorithm \cite{algorithm}), followed by  adding the coding vectors of the coded data chunks along the cycle to  coding matrix $G$ in Line~\ref{algorithm1:add-coding}. Subsequently,  Alg.~\ref{alg:greedy-vcg-coding}  removes the cycle from the present graph in Line~\ref{algorihtm1:remove}. The condition in Line~\ref{algorithm1:condition} guarantees that the maximum weight cycle in the present graph has a non-negative weight. 
Finally, if client $c_i$ submitting valuation $\hat{v}_i \geq 1$ is not in those selected cycles, then Alg.~\ref{alg:greedy-vcg-coding} adds the coding vector of data chunk $w_i$ to $G$ in Line~\ref{algorithm1:serve2}. 

Let $G_{\text{alg~\ref{alg:greedy-vcg-coding}}}$ be the coding matrix produced by Alg. \ref{alg:greedy-vcg-coding}  and  let $G^*$ be a sparse coding matrix maximizing
function $w(\hat{\mathbf{V}}, \hat{\mathbf{H}}, G)$. The next theorem analyzes the \emph{approximation ratio} $\frac{w(\hat{\mathbf{V}}, \hat{\mathbf{H}}, G^*)}{w(\hat{\mathbf{V}}, \hat{\mathbf{H}}, G_{\text{alg~\ref{alg:greedy-vcg-coding}}})}$ of Alg.~\ref{alg:greedy-vcg-coding}. 

\begin{theorem} \label{theorem:apx-ratio}
The approximation ratio (with respect to an optimal sparse coding matrix) of Alg.~\ref{alg:greedy-vcg-coding} is  the maximum cycle length in  weighted dependency graph $\mathcal{G}(\mathbf{\Lambda}, \mathbf{A}, \mathbf{\Gamma})$. 
\end{theorem} 

\begin{proof}
See Appendix \ref{appendix:theorem:apx-ratio}.
\end{proof} 

Because of Theorem~\ref{theorem:apx-ratio}, we  refer to Alg.~\ref{alg:greedy-vcg-coding} as $\max|C|$-approximate coding scheme. Next, we show that  applying $\max|C|$-approximate coding scheme to solve Eqs.~(\ref{eq:vcg-coding}) and ~(\ref{eq:vcg-pricing}) is no longer a truthful mechanism.

%
%
%
%
%

\begin{figure}
	\begin{center}
		\includegraphics[width=.6\textwidth]{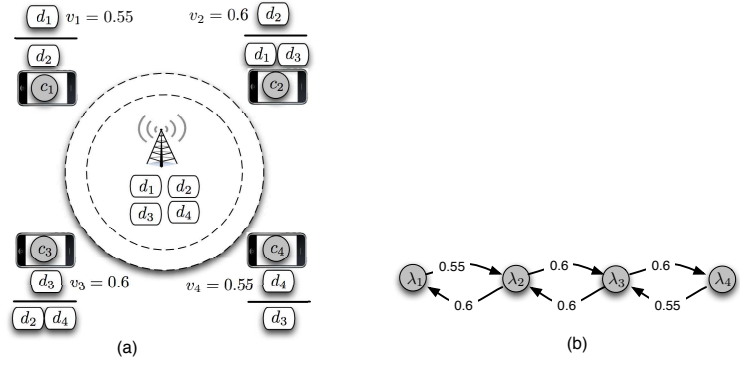}
		\caption{(a) Counter-example of the truthfulness property when we use Alg.~\ref{alg:greedy-vcg-coding} to solve Eqs.~(\ref{eq:vcg-coding}) and~(\ref{eq:vcg-pricing}); (b) The weighted dependency graph.}
		\label{fig:counter-ex}
	\end{center}
\end{figure}
\begin{example} \label{ex:counter}
Look at  Fig.~\ref{fig:counter-ex}.  First, suppose that   all clients  submit  the true valuations of the data chunks they want. Then,  Alg.~\ref{alg:greedy-vcg-coding}  produces  $d_2+d_3$ along cycle $(\lambda_2,\lambda_3)$ in Fig.~\ref{fig:counter-ex}(b). In this case, client $c_1$ has zero utility. Second,  suppose that client $c_1$ submits $\hat{v}_1=0.7$ but other clients submit the true valuations of the data chunks they want. Then, Alg.~\ref{alg:greedy-vcg-coding}  produces  $d_1+d_2$ and $d_3+d_4$ along cycles $(\lambda_1,\lambda_2)$ and $(\lambda_3, \lambda_4)$, respectively. By solving Eq.~(\ref{eq:vcg-pricing}) with Alg.~\ref{alg:greedy-vcg-coding}, client $c_1$ is charged  $p_1=(0.6+0.6-1)-(0.6+0.6+0.55-2)=0.45$. In this case, client $c_1$ has utility $0.55-0.45=0.1$.   Client $c_1$ can obtain a higher utility  by lying about the valuation of data chunk $w_i$. 
\end{example}

To address the issue  in the above example, we propose a payment scheme in Alg.~\ref{alg:greedy-vcg-pricing} so that  the joint design of Algs. \ref{alg:greedy-vcg-coding} and~\ref{alg:greedy-vcg-pricing} is a truthful mechanism. 
The underlying idea of Alg.~\ref{alg:greedy-vcg-pricing} is to calculate  threshold $\bar{v}_i(\hat{\mathbf{V}}_{-i},\hat{\mathbf{H}})$ for each client $c_i$ (with $\mathbf{1}_i(\hat{H}_i,G)=1$ under coding matrix $G$ constructed by Alg.~\ref{alg:greedy-vcg-coding}) as payment $p_i$. To that end, Alg.~\ref{alg:greedy-vcg-pricing} constructs weighted dependency graph $\mathcal{G}(\mathbf{\Lambda}, \mathbf{A}, \mathbf{\Gamma})$ in Line~\ref{alg:greedy-vcg-pricing:graph}; moreover, Alg.~\ref{alg:greedy-vcg-pricing} associates each arc $a \in \mathbf{A}$ with an arc cost $\zeta_a$ in Line~\ref{alg:greedy-vcg-pricing:cost}. Note that Alg.~\ref{alg:greedy-vcg-pricing}  defines the arc costs in a different way from Alg.~\ref{alg:greedy-vcg-coding}; precisely, Alg.~\ref{alg:greedy-vcg-pricing} associates each outgoing arc from vertex $\lambda_i$ with the cost of one unit (i.e., assuming  valuation $\hat{v}_i=0$). Then, Alg.~\ref{alg:greedy-vcg-pricing} calculates the difference of the cycle costs between  cycle $C_1$ (in Line~\ref{alg:greedy-vcg-pricing:c1}) and cycle $C_2$ (in Line~\ref{alg:greedy-vcg-pricing:c2}), where cycle $C_1$  has the \textit{globally} maximum weight  but cycle $C_2$ has the \textit{locally} maximum weight among those cycles  containing vertex $\lambda_i$.  While the value of $1-\zeta(C_1)$ (see Eq.~(\ref{eq:cycle-weight-2}))  is analogous to  the first term of Eq.~(\ref{eq:vcg-pricing}), that of $1-\zeta(C_2)$  is analogous to the second term of Eq.~(\ref{eq:vcg-pricing}). Thus, the difference $\zeta(C_2)-\zeta(C_1)$ of the cycle costs in Line~\ref{alg:greedy-vcg-pricing:s1} for each iteration is the minimum valuation $\hat{v}_i$ submitted by client $c_i$ such that a cycle  containing  vertex $\lambda_i$ can be selected by Line~\ref{algorithm1:greedy} of Alg.~\ref{alg:greedy-vcg-coding} in that iteration. Then, Alg.~\ref{alg:greedy-vcg-pricing} identifies  threshold $\bar{v}_i(\hat{\mathbf{V}}_{-i},\hat{\mathbf{H}})$ by searching for the minimum  among all iterations in Line~\ref{alg:greedy-vcg-pricing:s1} along with the initial valuation of $p_i$ being 1 as in Line~\ref{alg:greedy-vcg-pricing:start} (because each client $c_i$ can recover the data chunk it wants when submitting $\hat{v}_i \geq 1$).  
%


\begin{algorithm}[t]
\SetCommentSty{text}
\SetAlgoLined 
\SetKwFunction{Union}{Union}\SetKwFunction{FindCompress}{FindCompress} \SetKwInOut{Input}{input}\SetKwInOut{Output}{output}

\Input{Valuation set $\hat{\mathbf{V}}$,  side information set $\hat{\mathbf{H}}$, and client $c_i$ with $\mathbf{1}_i(\hat{H}_i, G)=1$ under coding matrix $G$ constructed by Alg.~\ref{alg:greedy-vcg-coding}.}
\Output{Payment $p_i$ for client $c_i$.}



Construct  weighted dependency graph $\mathcal{G}(\mathbf{\Lambda}, \mathbf{A}, \mathbf{\Gamma})$\;\label{alg:greedy-vcg-pricing:graph}
Associate each arc $a \in \mathbf{A}$ with arc cost $\zeta_{a}$:\label{alg:greedy-vcg-pricing:cost}
\begin{align*}
\zeta_a=\left\{
\begin{array}{ll}
1 & \text{if $a=(\lambda_i,\lambda_j)$ for some $j$};\\
1-[\gamma_a]^+_1 & \text{if $a \neq (\lambda_i,\lambda_j)$ for all $j$}.
\end{array}
\right.
\end{align*}
  
 
	$p_i\leftarrow 1$\; \label{alg:greedy-vcg-pricing:start}
	
\While{there exists a cycle in the present graph $\mathcal{G}(\mathbf{\Lambda}, \mathbf{A}, \mathbf{\Gamma})$ whose  cost is less than or equal to one, and there exists a cycle containing vertex $\lambda_i$
 \label{algorithm2:iteration}}  {
	Find a cycle $C_1=\arg\min_{C} \zeta(C)$ in the present graph $\mathcal{G}(\mathbf{\Lambda}, \mathbf{A}, \mathbf{\Gamma})$ for minimizing the cycle cost\; \label{alg:greedy-vcg-pricing:c1}
	
	Find a cycle $C_2 = \arg\min_{C \cap \lambda_i \neq \emptyset} \zeta(C)$ in the present graph $\mathcal{G}(\mathbf{\Lambda}, \mathbf{A}, \mathbf{\Gamma})$ that contains vertex $\lambda_i$ and minimizes the cycle cost among those cycles containing vertex $\lambda_i$\; \label{alg:greedy-vcg-pricing:c2}

	$p_i\leftarrow\min\{p_i,\zeta(C_2)-\zeta(C_1)\}$\; \label{alg:greedy-vcg-pricing:s1}
	
	Remove all vertices in cycle $C_1$ and their incident arcs  from the present graph $\mathcal{G}(\mathbf{\Lambda}, \mathbf{A}, \mathbf{\Gamma})$\;\label{alg:greedy-vcg-pricing:remove}
}



\caption{Payment scheme for those clients with $\mathbf{1}_i(\hat{H}_i, G)=1$ under coding matrix $G$ constructed by  Alg.~\ref{alg:greedy-vcg-coding}.}

\label{alg:greedy-vcg-pricing}
\end{algorithm}

By verifying the four conditions in Theorem~\ref{theorem:truthfulness}, the next theorem shows that the joint Algs.~\ref{alg:greedy-vcg-coding} and \ref{alg:greedy-vcg-pricing} is a truthful mechanism. 
\begin{theorem} \label{theorme:alg1-truthful}
In the multiple unicast scenario, the mechanism consisting of the coding scheme in Alg.~\ref{alg:greedy-vcg-coding} and the payment scheme in Alg.~\ref{alg:greedy-vcg-pricing} is a truthful mechanism. 
\end{theorem}
\begin{proof}
See Appendix \ref{appendix:alg1-truthful}.
\end{proof}

\subsection{$\sqrt{n}$-approximate truthful mechanism} \label{subsection:approximation-alg2}

This section proposes another approximate sparse coding scheme and its corresponding payment scheme for guaranteeing the truthfulness. The approximate coding scheme modifies the previously proposed Alg.~\ref{alg:greedy-vcg-coding}. The modified approximate coding scheme  substitutes Line~\ref{algorithm1:greedy} of Alg.~\ref{alg:greedy-vcg-coding} (i.e., identifying  a  cycle for maximizing  cycle weight $\gamma(C)$) by  identifying a cycle $C$ for maximizing  $\frac{\gamma(C)}{\sqrt{|C|}}$. The underlying idea is to maximize cycle weight $\gamma(C)$ (as in Alg.~\ref{alg:greedy-vcg-coding}) and at the same time to minimize the number of transmissions (because shorter cycle lengths $|C|$ can yield more cycles). To that end, we propose Alg. \ref{alg:max-sqrt-weight}  for obtaining such a cycle in a weighted dependency graph.   Line \ref{alg-max-sqrt:fix-i} of Alg.~\ref{alg:max-sqrt-weight}  searches for a cycle for maximizing cycle weight  subject to the cycle length being no more than $i$. Then,  Line~\ref{alg-max-sqrt:opt} of Alg.~\ref{alg:max-sqrt-weight} can identify cycle $C$ for maximizing $\frac{\gamma(C)}{\sqrt{|C|}}$ subject to  the cycle length being no more than $i$; in particular, Line~\ref{alg-max-sqrt:opt} can identify cycle $C$ for maximizing $\frac{\gamma(C)}{\sqrt{|C|}}$ in the last iteration.  See Appendix \ref{appendix:lemma:maximum-valuation-cycle} for carefully justifying the correctness of Alg.~\ref{alg:max-sqrt-weight}.

\begin{algorithm}[t]
\SetCommentSty{text}
\SetAlgoLined 
\SetKwFunction{Union}{Union}\SetKwFunction{FindCompress}{FindCompress} \SetKwInOut{Input}{input}\SetKwInOut{Output}{output}

\Input{Weight dependency graph $\mathcal{G}(\mathbf{\Lambda}, \mathbf{A}, \mathbf{\Gamma})$.}
\Output{Cycle $C$  maximizing $\frac{\gamma(C)}{\sqrt{|C|}}$.}

$C \leftarrow \emptyset$\;

\For{$i \leftarrow 2$ \KwTo $n$}{
Find a cycle $C'=\arg\min_{C''} \zeta(C'')$ subject to $|C''| \leq i$\; \label{alg-max-sqrt:fix-i}

 $C \leftarrow \arg\max_{C' \text{\,or\,} C} \{\frac{\gamma(C')}{\sqrt{|C'|}}, \frac{\gamma(C)}{\sqrt{|C|}}\}$\; \label{alg-max-sqrt:opt}
} 

\caption{Identifying cycle $C$ for maximizing $\frac{\gamma(C)}{\sqrt{|C|}}$.}

\label{alg:max-sqrt-weight}
\end{algorithm}

The next theorem provides the approximation ratio of the modified approximation algorithm. 
\begin{theorem} \label{theorem:sqrt-approximation}
In the multiple unicast scenario, substituting Line \ref{algorithm1:greedy}  of Alg. \ref{alg:greedy-vcg-coding} by identifying a cycle $C$ for maximizing $\frac{\gamma(C)}{\sqrt{|C|}}$ yields the approximation ratio (with respect to the welfare yielded by an optimal sparse coding matrix) of $\sqrt{n}$.
\end{theorem}
\begin{proof}
See Appendix \ref{appendix:theorem:sqrt-approximation}.
\end{proof}

Because of Theorem~\ref{theorem:sqrt-approximation}, we refer to the coding scheme modified from Alg.~\ref{alg:greedy-vcg-coding} as $\sqrt{n}$-approximate coding scheme. To the best of our knowledge, \cite{cycle-packing} developed the best approximation algorithm for a cycle packing problem (identifying the maximum number of disjoint cycles), which is a special case of our problem when $\hat{v}_i=1$ for all~$i$. 
That paper  showed that the approximation ratio of that algorithm is $\sqrt{n}$;	furthermore, it conjectured that $\sqrt{n}$ is the best approximation ratio for that cycle packing problem.

Moreover, following Theorem \ref{theorme:alg1-truthful}, we can establish the truthfulness as follows. 

\begin{theorem}
Substituting  Line \ref{algorithm1:greedy}  of Alg. \ref{alg:greedy-vcg-coding} and  Lines~\ref{alg:greedy-vcg-pricing:c1} and~\ref{alg:greedy-vcg-pricing:c2} of Alg.~\ref{alg:greedy-vcg-pricing}  by identifying a cycle $C$ for maximizing $\frac{\gamma(C)}{\sqrt{|C|}}$ yields a truthful mechanism. 
\end{theorem}

\subsection{Complexities of the proposed coding schemes} \label{subsection:complexity}
This section investigates the computational complexities of the three proposed coding schemes for the multiple unicast scenario: 1) Alg. \ref{alg:poly-time}  for the set $\mathbf{G}$ of sparse and instantly decodable coding matrices; 2) $\max |C|$-approximate coding scheme (Alg. \ref{alg:greedy-vcg-coding}) for  for the set $\mathbf{G}$ of sparse  coding matrices; 3) $\sqrt{n}$-approximate coding scheme (modified Alg. \ref{alg:greedy-vcg-coding} along with Alg.~\ref{alg:max-sqrt-weight}) for the set $\mathbf{G}$ of sparse  coding matrices. All  three schemes are based on a weighted dependency graph. Constructing a weighted dependency graph takes $O(n^2)$ steps to check all pairs of clients. 

Regarding Alg.~\ref{alg:poly-time}, we can apply the Edmond's maximum weight matching algorithm \cite{galil1986efficient} to Line~\ref{alg:poly-time:max-matching}  of Alg. \ref{alg:poly-time}, whose complexity is $O(n^3)$. Then, the complexity of Alg. \ref{alg:poly-time} is $O(n^3)$. 

Regarding Alg.~\ref{alg:greedy-vcg-coding}, we can apply  the Floyd-Warshall algorithm \cite{algorithm} to Line~\ref{algorithm1:cost} for each iteration, whose complexity  is $O(n^3)$. Because there are at most $n/2$ cycles in a weighted dependency graph (i.e., at most $n/2$ iterations from Line \ref{algorithm1:condition}),  the complexity of $\max|C|$-approximate coding scheme is $O(n^4)$.

Regarding Alg.~\ref{alg:max-sqrt-weight}, we can apply the Bellman-Ford algorithm \cite{algorithm} to Line \ref{alg-max-sqrt:fix-i} for each iteration $i$, whose complexity is $O(i \cdot n^2)$. Hence, the complexity of Alg. \ref{alg:max-sqrt-weight} is $O(n^4)$; in turn, the complexity of the $\sqrt{n}$-approximate coding scheme is $O(n^5)$.

\subsection{Numerical results} \label{section:simulation}

This section numerically analyzes the proposed  coding schemes for the multiple unicast scenario via computer simulations, including Alg. \ref{alg:poly-time}, $\max|C|$-approximate coding scheme in Alg.~\ref{alg:greedy-vcg-coding}, and  $\sqrt{n}$-approximate coding scheme modified from Alg.~\ref{alg:greedy-vcg-coding}.

Fig.~\ref{fig:result12} simulates Alg. \ref{alg:poly-time} with the instant decoding scheme and the two approximate coding schemes with the general decoding scheme. The two sub-figures display the social welfare when each client has 3 and 6 data chunks, respectively, in its side information. The experiment setting is following: We simulate $n$ clients (x-axle) and set $\mathbf{D}=\{d_1, \cdots, d_n\}$, where client $c_i$ wants data chunk $d_i$. Valuation $v_i$ of data chunk $d_i$ is uniformly picked between 0 and 1.  The data chunks in side information $H_i$ of client $c_i$ is randomly selected from set $\mathbf{D} - \{d_i\}$. As a result of the proposed payment schemes as the incentives, we can guarantee that all clients submit their true information. Thus, we simulate the proposed coding schemes along with the true information (like no selfish clients). Note that under that setting, the proposed approximate coding schemes can approximate the \textit{true} social welfare in Eq.~(\ref{eq:social-welfare}). 
  All results are averaged over 500 simulation times.

From Fig.~\ref{fig:result12}, we can observe that even though both approximate coding schemes cannot achieve the maximum social welfare (in the set $\mathbf{G}$ of sparse coding matrices),  they still outperform Alg.~\ref{alg:poly-time} (that is optimal in the set of sparse and instantly decodable coding matrices). The result tells us that the proposed approximate coding schemes can take advantage of the general  decoding scheme.  Moreover, Table~\ref{table:difference} shows the difference between the social welfare yielded by  $\sqrt{n}$-approximate coding scheme in Alg.~\ref{alg:greedy-vcg-coding} and that yielded by Alg.~\ref{alg:poly-time} for various fixed numbers of clients and fixed sizes of side information.  From the table,  the difference when $|H_i|=6$ is almost more than that when $|H_i|=3$ (especially when the number of clients is larger like $n=34$, $40$, or $46$), i.e., the advantage from the general decoding scheme becomes more obvious when the clients have more data chunks as their side information. That would be because when clients have more data chunks as their side information, there are more cycles in the weighted dependency graph whose lengths are more than two.

\begin{table*}[!t]
	\footnotesize
	\caption{Differences of the social welfares}
	
	\begin{center}
		\begin{tabular}{|c|ccccccc|}
			\hline
		the number $n$ of client & 10 & 16 & 22 & 28 & 34 & 40 & 46\\
		\hline
		$|H_i|=3$ & 0.03 & 0.1 & 0.12 &  0.17 & 0.17 & 0.2 & 0.21\\ 
		\hline
		$|H_i|=6$ & 0 & 0.1 & 0.12 & 0.2 & 0.22 & 0.38 & 0.39\\
		\hline
		\end{tabular}
		\label{table:difference}
	\end{center}
\end{table*}


\begin{figure}[t]
	\begin{minipage}{.45\textwidth}
		\centering
			\includegraphics[width=\textwidth]{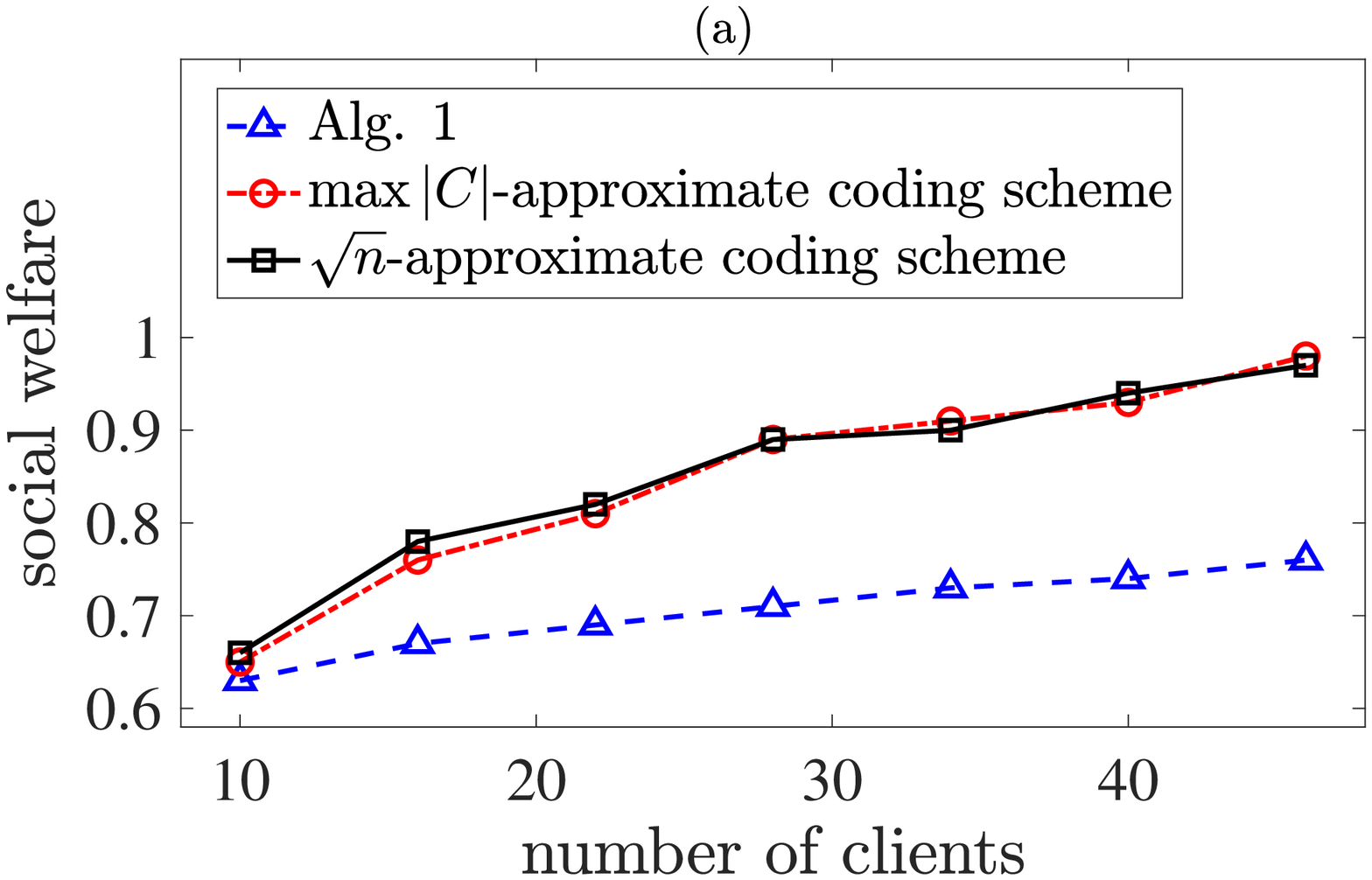}
	\end{minipage}\hfill
	\begin{minipage}{.45\textwidth}
	\centering
			\includegraphics[width=\textwidth]{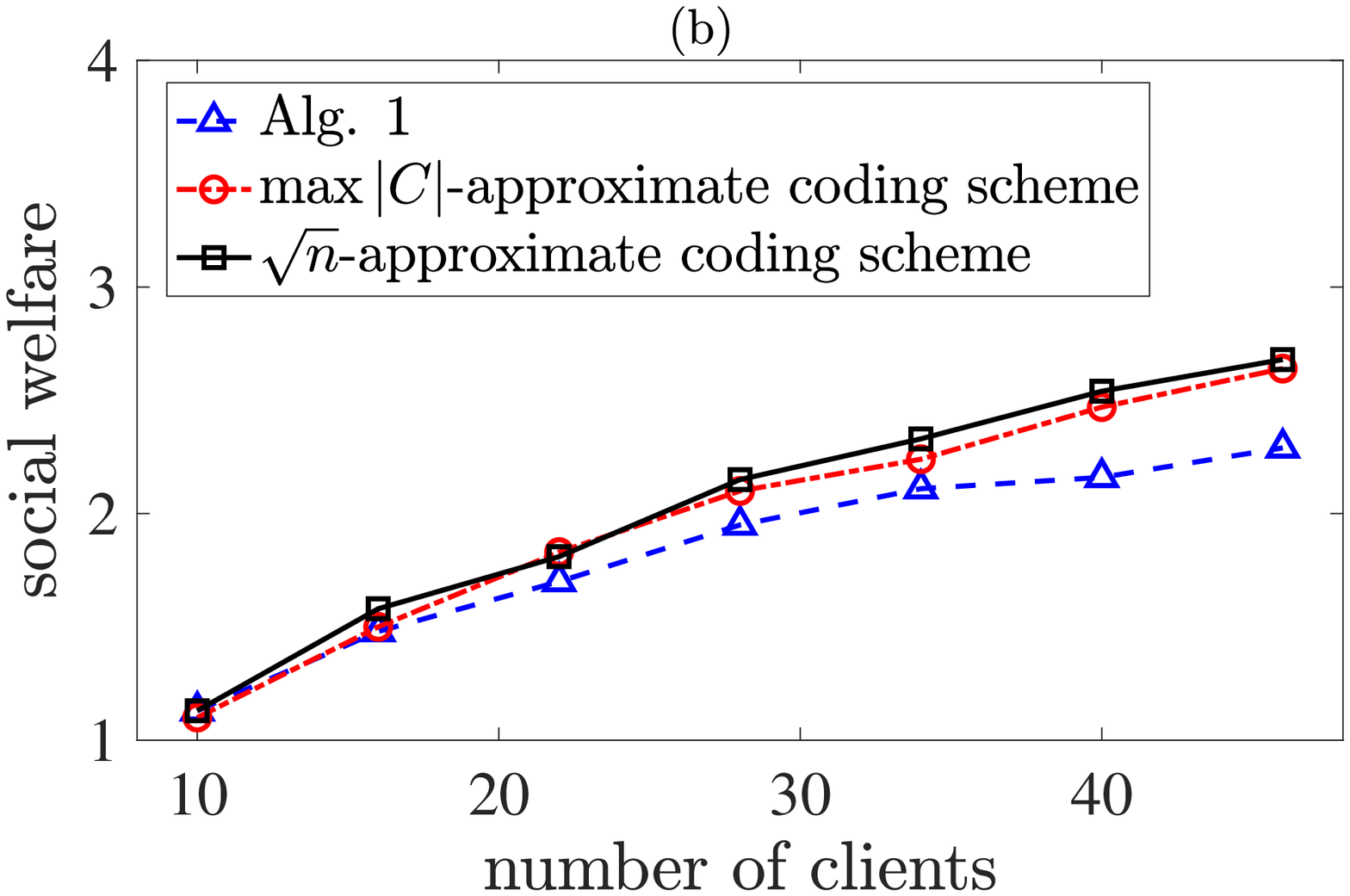}
	\end{minipage}\hfill
		\caption{Social welfare of Alg.~\ref{alg:poly-time}, $\max|C|$-approximate coding scheme in Alg.~\ref{alg:greedy-vcg-coding}, and $\sqrt{n}$-approximate coding scheme modified from Alg.~\ref{alg:greedy-vcg-coding}: (a) each client has 3 data chunks in its side information; (b) each client has 6 data chunks in its side information;}
	\label{fig:result12}	
\end{figure}
\begin{figure}[t]
	\begin{minipage}{.45\textwidth}
		\centering
		\includegraphics[width=\textwidth]{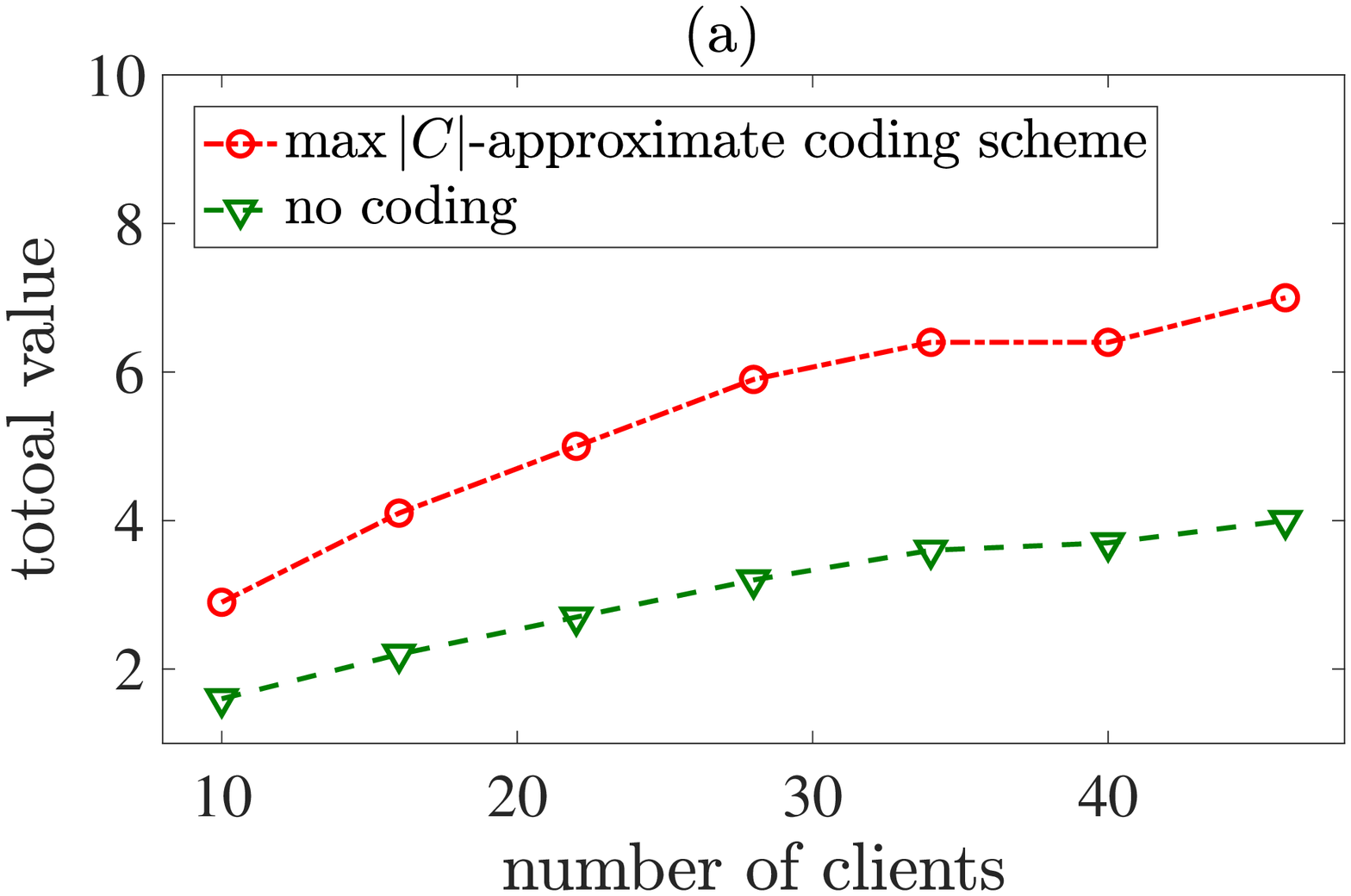}
	\end{minipage}\hfill
	\begin{minipage}{.45\textwidth}
		\centering
		\includegraphics[width=\textwidth]{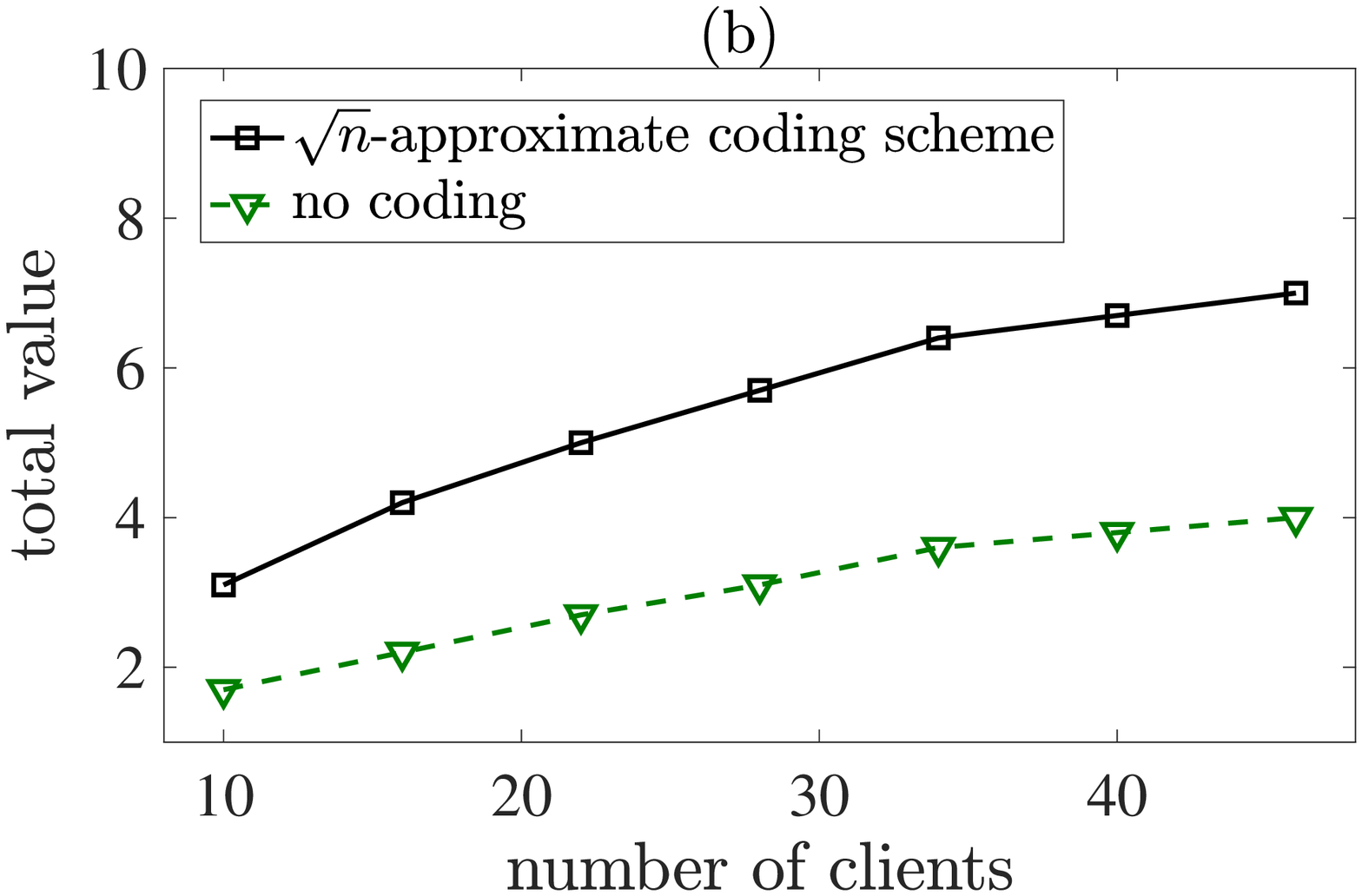}
	\end{minipage}\hfill
	\caption{Benefit from the approximate coding schemes when each client has 3 data chunks in its side information.}
	\label{fig:result34}	
\end{figure}

\begin{figure}[t]
	\begin{minipage}{.45\textwidth}
		\centering
		\includegraphics[width=\textwidth]{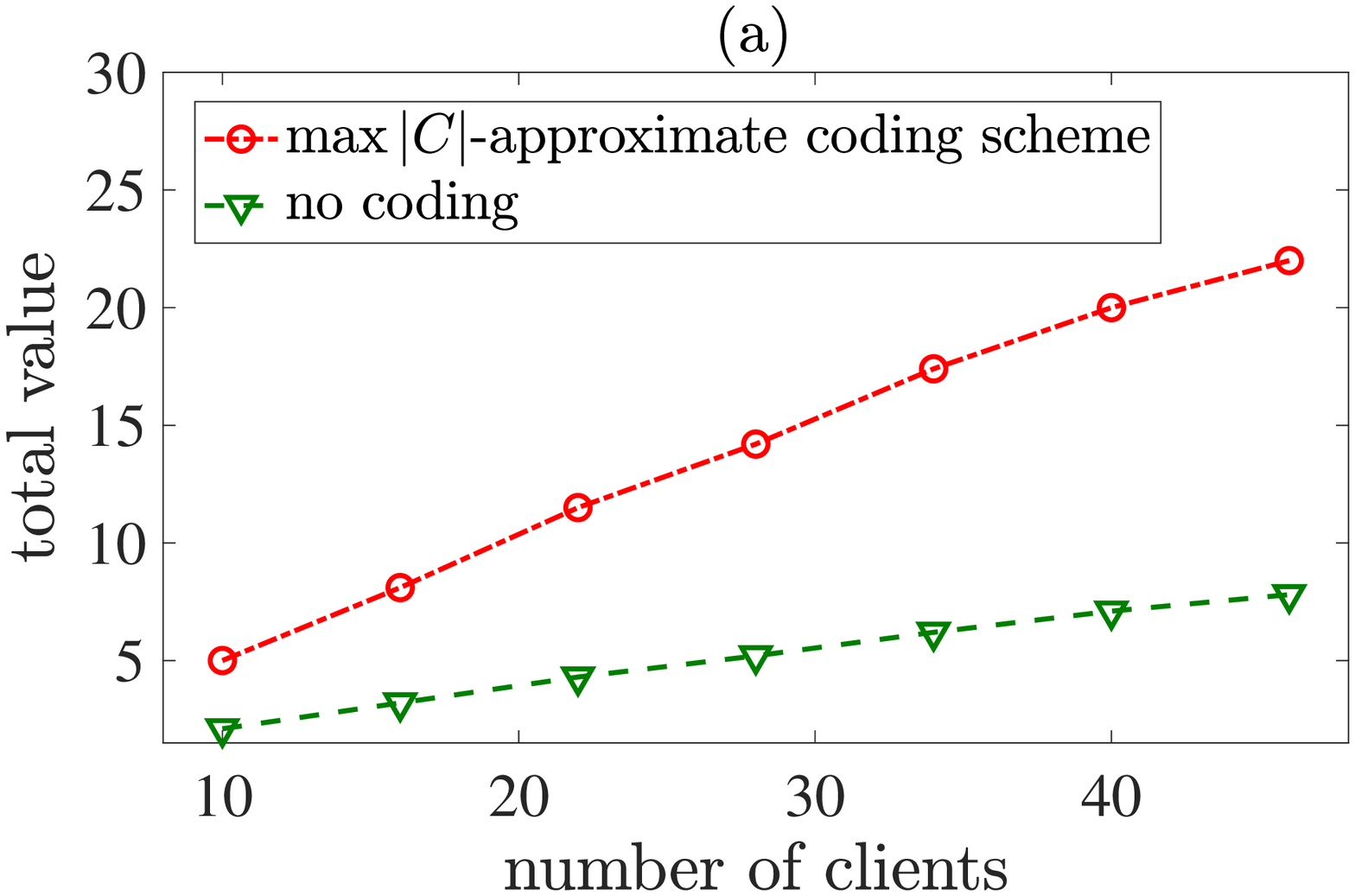}
	\end{minipage}\hfill
	\begin{minipage}{.45\textwidth}
		\centering
		\includegraphics[width=\textwidth]{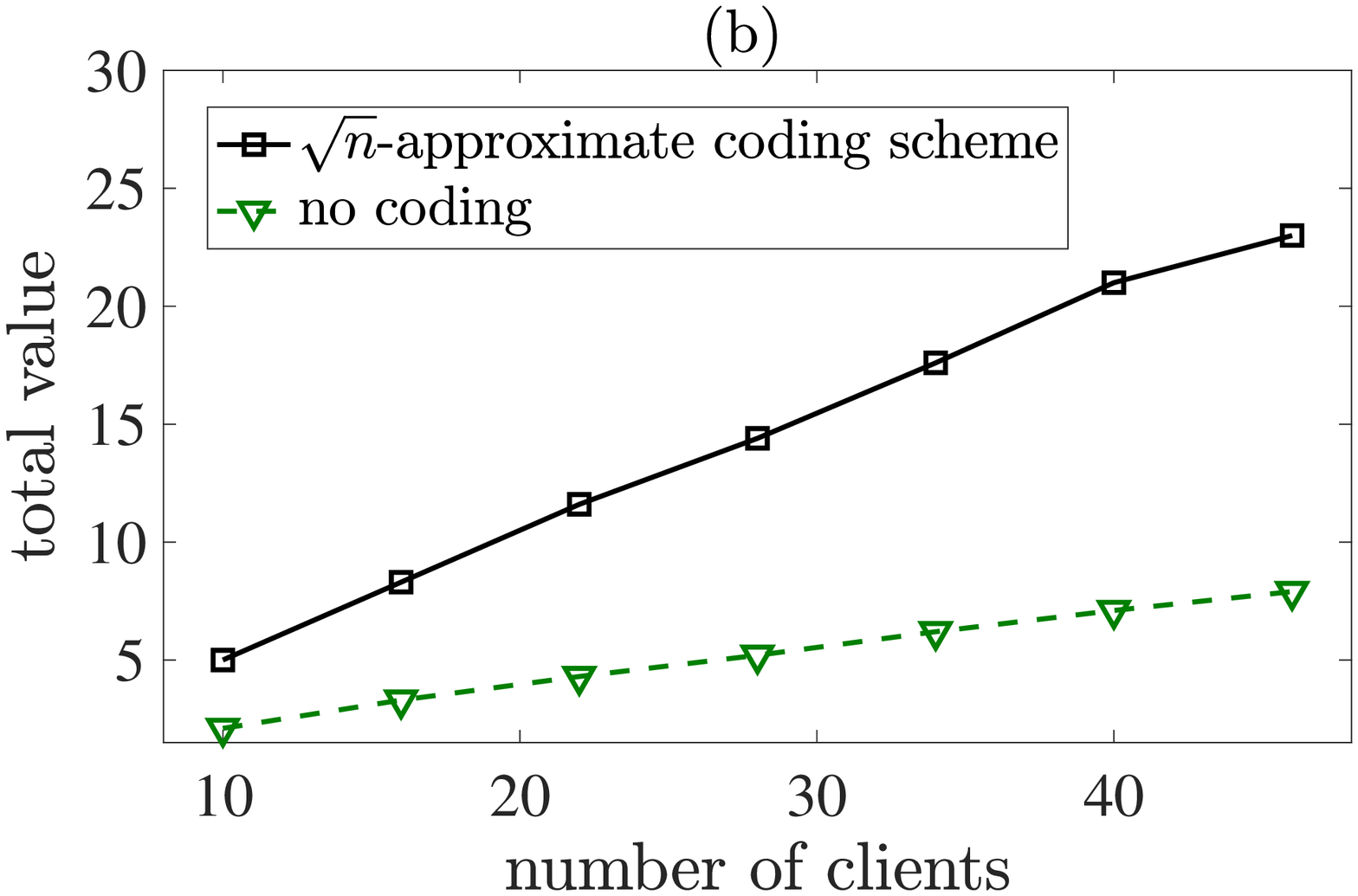}
	\end{minipage}\hfill
	\caption{Benefit from the approximate coding schemes  when each client has 6 data chunks in its side information.}
	\label{fig:result56}	
\end{figure}

While the proposed approximate coding schemes cannot be optimal in the set $\mathbf{G}$ of sparse coding matrices, we want to further validate them via simulations.  Figs.~\ref{fig:result34} and~\ref{fig:result56} display the total valuation $\sum_{i=1}^n v_i \cdot \mathbf{1}_i(H_i, G)$ of the data chunks that can be recovered. For a fair comparison of the total valuation, we have to fix the number of transmissions from the server for various coding schemes. Thus, we cannot compare the total valuation yielded by the approximate coding schemes with that yielded by Alg.~\ref{alg:poly-time} (as they cause different numbers of transmissions). Instead, Figs.~\ref{fig:result34} and~\ref{fig:result56} compare the approximate coding schemes with the best ``no coding'' scheme.  The results for the ``no coding'' scheme in Figs.~\ref{fig:result34} and~\ref{fig:result56} are obtained in the following way: we first obtain the  number  $\eta(G)$ of transmissions incurred by a given approximate coding scheme, and then calculate the sum of the top  $\eta(G)$ valuations of the data chunks in set $\mathbf{D}$ (which is the maximum total valuation when the server transmits $\eta(G)$ uncoded data chunks).  From Figs.~\ref{fig:result12} - \ref{fig:result56}, we can observe that both approximate coding schemes can improve the total valuation over the best uncoded transmission scheme.

\section{Inapproximability results for the multiple multicast scenario} \label{section:multiple-multicast}
 
Thus far, we analyzed the multiple unicast scenario. This section  analyzes the multiple multicast scenario.  Given a graph $\mathcal{G}(\mathbf{\Lambda},\mathbf{E})$ (with vertex set $\mathbf{\Lambda}$ and edge set $\mathbf{E}$) of the independent set  problem, we construct an instance of our problem as follows. For each vertex $\lambda \in \mathbf{\Lambda}$ and edge $e \in \mathbf{E}$, we construct data chunks $d_{\lambda}$ and $d_e$. The data chunk set $\mathbf{D}$  consists of $d_{\lambda}$ and $d_e$ for all $\lambda \in \mathbf{\Lambda}$ and $e \in \mathbf{E}$. For each edge $e=(x,y) \in \mathbf{E}$, we construct three clients $c_{e,1}, c_{e,2}, c_{e,3}$ such that
\begin{itemize}
	\item $w_{e,1}=d_e$, $\hat{H}_{e,1}=\{d_x, d_y\}$, $\hat{v}_{e,1}=1$,
	\item $w_{e,2}=d_x$, $\hat{H}_{e,2}=\{d_e\}$, $\hat{v}_{e,2}=\frac{1}{\text{deg}(x)}$,
	\item $w_{e,3}=d_y$, $\hat{H}_{e,3}=\{d_e\}$, $\hat{v}_{e,3}=\frac{1}{\text{deg}(y)}$,
\end{itemize}
where $\text{deg}(\lambda)$ is the number of edges that are incident to vertex $\lambda \in \mathbf{\Lambda}$. With the reduction,  Appendix~\ref{appendix:theorem:general-arg-hard-multicast} shows that the combinatorial optimization problem in Eq.~(\ref{eq:vcg-coding}) equivalently becomes the independent set problem (for both the set $\mathbf{G}$ of sparse and instantly decodable coding matrices and the set $\mathbf{G}$ of sparse and general decodable coding matrices). Thus, the social welfare under our reduction is at most the number $|\boldsymbol{\Lambda}|$ of vertices (i.e., the approximation ratio is at most $|\boldsymbol{\Lambda}|$). Following the hardness result \cite{clq-hard} of the independent set problem, the combinatorial optimization problem in Eq.~(\ref{eq:vcg-coding}) in the multiple multicast scenario is NP-hard and NP-hard to approximate within the factor of $|\boldsymbol{\Lambda}|^{1-\epsilon}$ for any constant $\epsilon>0$.

%
%
%
%
%
%
%

\section{Concluding remarks}\label{sec:conclusion}
This paper treated the index coding setting in the presence of selfish clients. We proposed a sufficient condition for truthful mechanisms (i.e., joint coding and payment schemes). Leveraging the proposed condition, we proposed truthful mechanisms, including the VCG-based mechanism and some polynomial-time mechanisms. While the VCG-based mechanism can maximize the social welfare, other polynomial-time mechanisms can either maximize the social welfare  or approximate it with provable approximation ratios within some restricted  sets of coding matrices. Following are some remarks on the assumptions made in the paper and possible extensions.

\begin{itemize}
	\item This paper supposed that each client wants a single data chunk. Consider the case when a	client wants a bundle of more than one data chunk. If the total valuation of different data chunks is the sum of their individual valuations (i.e., the client treats each data chunks in the bundle separately), then the client can be substituted by multiple clients that each wants a single data chunk in the bundle and has the same side information. In this context, the proposed mechanisms and their results hold for that case. The more general case when there are different valuations for different subsets of the bundle would be an interesting extension.  
	\item This paper assumed scalar-linear coding schemes. We made that assumption  for consistency.  In fact, while those polynomial-time mechanisms (with sparse coding schemes) need that assumption, Theorem~\ref{theorem:truthfulness} and the proposed VCG-based mechanism do not. 
	\item This paper developed polynomial-time mechanisms with sparse coding schemes. To develop a mechanism that can encode more than two data chunks is an interesting extension. A possible solution would be to construct an approximation algorithm that can identify a clique cover in a weighted dependency graph. Then, for each of the cliques, the server can broadcast a single XOR-coded data chunk combining all data chunks associated with the clique (like \cite{birk1998informed}).
	\item The inapproximability results for the multiple multicast scenario in the paper is based on the proposed VCG-based mechanism. It is a promising future work to develop a polynomial-time mechanism for the multiple multicast scenario beyond the scope of the  VCG-based mechanism.
	\item The payment schemes proposed in the paper can motivate each client to submit its true information such that the server can maximize or approximate the true social welfare. However, the payment collected by the server do not always cover the transmission cost. For example, the proposed VCG-based mechanism satisfies the fourth condition in Theorem~\ref{theorem:truthfulness}. Though that issue might be hard to avoid according to \cite[Section 9.3.5.5]{nisan2007algorithmic}, that is still an interesting future work. 
	
\end{itemize}


\appendices

\section{Proof of Theorem  \ref{theorem:truthfulness}} \label{appendix:theorem:truthfulness}
Consider a fixed valuation set $\hat{\mathbf{V}}_{-i}$ and side information set $\hat{\mathbf{H}}_{-i}$. Next, we consider two cases as follows.
\begin{enumerate}

	\item \textit{The coding scheme can construct a coding matrix $G$ such that $\mathbf{1}_i(\hat{H}_i, G)=0$ but $\mathbf{1}_i(H_i, G)=1$ for some $\hat{v}_i$ and $\hat{H}_i$}: First, suppose that client $c_i$ submits the true valuation $v_i$ and the complete side information $H_i$. By the first and fourth conditions, client $c_i$ can recover data chunk $w_i$ if $v_i>0$ and cannot if $v_i=0$. By the second and the fourth condition, client is charged 0 if it can recover $w_i$. Thus, client $c_i$ has utility $v_i$.   Second, suppose that client $c_i$ submits a valuation $\hat{v}_i\neq v_i$ or side information $\hat{H}_i \subset H_i$. By the definition of the utility, client $c_i$ has utility $(v_i-p_i)\cdot \mathbf{1}_i(H_i, G)\leq v_i$ or $(v_i-p_i)\cdot \mathbf{1}_i(\hat{H}_i, G)\leq v_i$. Thus, client $c_i$ can maximize its utility by submitting both the true valuation and the complete side information.
	\item \textit{The coding scheme cannot construct a coding matrix $G$ such that $\mathbf{1}_i(\hat{H}_i, G)=0$ but $\mathbf{1}_i(H_i, G)=1$ for all $\hat{v}_i$ and $\hat{H}_i$}: Note that, under the condition of this case, we have $\mathbf{1}_i(H_i, G)=0$ if $\mathbf{1}_i(\hat{H}_i, G)=0$. Moreover,  we also have $\mathbf{1}_i(H_i, G)=1$ if $\mathbf{1}_i(\hat{H}_i, G)=1$, because of  $\hat{H}_i \subseteq H_i$. In summary, this case has $\mathbf{1}_i(H_i, G)=\mathbf{1}_i(\hat{H}_i, G)$  for all $\hat{H}_i \subseteq H_i$ and all possible coding matrices $G$.
	Then,  we claim that \textit{for a fixed $\hat{H}_i$, client $c_i$ can maximize its utility  by submitting the true valuation $v_i$ of  data chunk $w_i$.} To prove that claim, we  consider three cases as follows.
	
	\begin{enumerate}
		\item  \textit{$v_i > \bar{v}_i(\hat{\mathbf{V}}_{-i},\hat{\mathbf{H}})$}:  First, suppose that client $c_i$ submits the true valuation $v_i$ of data chunk $w_i$. By the first and the second conditions,  client $c_i$ can recover data chunk $w_i$ (because $\mathbf{1}_i(H_i, G)=\mathbf{1}_i(\hat{H}_i, G)=1$) and is charged  threshold $\bar{v}_i(\hat{\mathbf{V}}_{-i},\hat{\mathbf{H}})$. Thus, client $c_i$ has utility   $v_i-\bar{v}_i(\hat{\mathbf{V}}_{-i},\hat{\mathbf{H}}) > 0$.
		Second, suppose that  client $c_i$  submits valuation $\hat{v}_i \neq v_i$ of data chunk $w_i$.  By the first and second conditions, if $\mathbf{1}_i(\hat{H}_i, G)=1$ for some coding matrix $G$,  then client $c_i$ has utility   $v_i-\bar{v}_i(\hat{\mathbf{V}}_{-i},\hat{\mathbf{H}})$; if $\mathbf{1}_i(\hat{H}_i, G)=0$,  then client $c_i$ has zero utility.  In summary, client $c_i$ can maximize its utility  by submitting the true valuation $v_i$ of data chunk $w_i$. 
		\item \textit{$v_i < \bar{v}_i(\hat{\mathbf{V}}_{-i},\hat{\mathbf{H}})$}: First, suppose that client $c_i$ submits the true valuation $v_i$ of data chunk $w_i$. By the first  condition, client $c_i$ has  zero utility. Second, suppose that  client $c_i$ submits valuation $\hat{v}_i \neq v_i$ of data chunk $w_i$. By the first and second conditions, if $\mathbf{1}_i(\hat{H}_i, G)=1$ for some coding matrix $G$, then  client $c_i$ has utility $v_i-\bar{v}_i(\hat{\mathbf{V}}_{-i},\hat{\mathbf{H}})<0$; if $\mathbf{1}_i(\hat{H}_i, G)=0$, then  client $c_i$ has zero utility. In summary,  	client $c_i$ can maximize its utility  by submitting the true valuation $v_i$ of data chunk $w_i$.
		\item \textit{$v_i = \bar{v}_i(\hat{\mathbf{V}}_{-i},\hat{\mathbf{H}})$}: First, suppose that client $c_i$ submits the true valuation $v_i$ of data chunk $w_i$. By the first and second conditions, client $c_i$ has zero utility  for either $\mathbf{1}_i(\hat{H}_i, G)=1$ or $\mathbf{1}_i(\hat{H}_i, G)=0$, for all possible coding matrices $G$. Second, suppose that  client $c_i$ submits valuation $\hat{v}_i \neq v_i$ of data chunk $w_i$.  By the first and second conditions, if $\mathbf{1}_i(\hat{H}_i, G)=1$ for some coding matrix $G$, then  client $c_i$ has  utility  $v_i-\bar{v}_i(\hat{\mathbf{V}}_{-i},\hat{\mathbf{H}})=0$; if $\mathbf{1}_i(\hat{H}_i, G)=0$, then  client $c_i$ has zero utility. In summary, client $c_i$ can maximize its utility  by submitting the true valuation $v_i$ of data chunk $w_i$.
	\end{enumerate}

	According to the above claim, we can suppose that client $c_i$ submits its true valuation $v_i$. It remains to show that, \textit{client $c_i$ can maximize its utility  by submitting its complete side information $H_i$.} To prove that claim, we  consider three cases as follows.
	
	\begin{enumerate}
		
		\item  \textit{$v_i > \bar{v}_i(\hat{\mathbf{V}}_{-i},\{H_i, \hat{\mathbf{H}}_{-i}\})$}: First, suppose that client $c_i$ submits the complete side information $H_i$. Client $c_i$ has utility $v_i-\bar{v}_i(\hat{\mathbf{V}}_{-i},\{H_i, \hat{\mathbf{H}}_{-i}\}) > 0$. Second, suppose that  client $c_i$  submits incomplete side information $\hat{H}_i \subset H_i$. If $\mathbf{1}_i(\hat{H}_i, G)=1$ for some coding matrix $G$, then  client $c_i$ has utility $v_i-\bar{v}_i(\hat{\mathbf{V}}_{-i},\{\hat{H}_i, \hat{\mathbf{H}}_{-i}\}) \leq v_i-\bar{v}_i(\hat{\mathbf{V}}_{-i},\{H_i, \hat{\mathbf{H}}_{-i}\})$ by the third condition; if $\mathbf{1}_i(\hat{H}_i, G)=0$,  then client $c_i$ has zero utility.  In summary, client $c_i$ can maximize its utility  by submitting the complete side information $H_i$. 
		\item \textit{$v_i < \bar{v}_i(\hat{\mathbf{V}}_{-i},\{H_i, \hat{\mathbf{H}}_{-i}\})$}: First, suppose that client $c_i$ submits the complete side information $H_i$. Client $c_i$ has zero utility. Second, suppose that  client $c_i$  submits incomplete side information $\hat{H}_i \subset H_i$. By the third condition, we have $\bar{v}_i(\hat{\mathbf{V}}_{-i},\{\hat{H}_i,\hat{\mathbf{H}}_{-i}\}) \geq \bar{v}_i(\hat{\mathbf{V}}_{-i},\{H_i, \hat{\mathbf{H}}_{-i}\}) > v_i$. Client $c_i$ has zero utility. In summary, client $c_i$ can maximize its utility  by submitting the complete side information $H_i$. 
		\item \textit{$v_i = \bar{v}_i(\hat{\mathbf{V}}_{-i},\{H_i, \hat{\mathbf{H}}_{-i}\})$}: First, suppose that client $c_i$ submits the complete side information $H_i$. Client $c_i$ has  zero utility   for either $\mathbf{1}_i(H_i, G)=1$ or $\mathbf{1}_i(H_i, G)=0$, for all possible coding matrices $G$. Second, suppose that  client $c_i$  submits incomplete side information $\hat{H}_i \subset H_i$.  If $\mathbf{1}_i(\hat{H}_i, G)=1$ for some coding matrix $G$, then  client $c_i$ has  utility  $v_i-\bar{v}_i(\hat{\mathbf{V}}_{-i},\{\hat{H}_i, \hat{\mathbf{H}}_{-i}\}) \leq v_i-\bar{v}_i(\hat{\mathbf{V}}_{-i},\{H_i, \hat{\mathbf{H}}_{-i}\}) = 0$ by the third condition; if $\mathbf{1}_i(\hat{H}_i, G)=0$, then  client $c_i$ has zero utility. In summary, client $c_i$ can maximize its utility  by submitting the complete side information $H_i$. 
	\end{enumerate}
\end{enumerate}

\section{Proof of Theorem  \ref{theorem:truthful-coding-pricing}} \label{appendix:truthful-coding-pricing}

	Let   $\hat{\mathbf{G}}_i \subseteq \mathbf{G}$ be the set of coding matrices such that $\mathbf{1}_i(\hat{H}_i, G)=1$. Moreover,  let $\hat{\mathbf{G}}^c_i \subseteq \mathbf{G}$ be the set of coding matrices such that $\mathbf{1}_i(\hat{H}_i, G)=0$. Thus, we can express set $\mathbf{G}$ in Eq.~(\ref{eq:vcg-coding}) by $\mathbf{G}=\hat{\mathbf{G}}_i \cup \hat{\mathbf{G}}^c_i$. Moreover,  let  $\mathbf{G}_i \subseteq \mathbf{G}$ and $\mathbf{G}^c_i \subseteq \mathbf{G}$  be the sets of coding matrices $G$ such that $\mathbf{1}_i(H_i, G)=1$ and $\mathbf{1}_i(H_i, G)=0$, respectively. In addition, let $\mathbf{G}^*$ be the set of coding matrices that are solutions to Eq.~(\ref{eq:vcg-coding}).

Then, we  show that the VCG-based mechanism satisfies the four conditions in Theorem \ref{theorem:truthfulness} as follows. 

\begin{enumerate}
	\item \textit{The VCG-based coding scheme is a threshold-type coding scheme}: Consider a fixed  valuation set $\hat{\mathbf{V}}_{-i}$ and a fixed side information set $\hat{\mathbf{H}}$. First, suppose that while client $c_i$ submits a valuation $\hat{v}_i$ of data chunk $w_i$, the VCG-based coding scheme \textit{can} construct a coding matrix~$G$ such that $\mathbf{1}_i(\hat{H}_i, G)=1$. That is, 
	the solution set $\mathbf{G}^*$ to Eq.~(\ref{eq:vcg-coding}) either belongs to set $\hat{\mathbf{G}}_i$ or\footnote{In the either-or, the first case corresponds to the case when $\hat{v}_i>\bar{v}_i(\hat{\mathbf{V}}_{-i},\hat{\mathbf{H}})$ and the second one  corresponds to the case when when $\hat{v}_i=\bar{v}_i(\hat{\mathbf{V}}_{-i},\hat{\mathbf{H}})$.} has a non-empty intersection with both sets $\hat{\mathbf{G}}_i$ and $\hat{\mathbf{G}}^c_i$ (with the arbitrary tie-breaking rule).
		Let $G^* \in \mathbf{G}^* \cap \hat{\mathbf{G}}_i$ be a solution to Eq.~(\ref{eq:vcg-coding}) such that $\mathbf{1}_i(\hat{H}_i, G^*)=1$. Second, suppose that  client $c_i$ submits a valuation $\tilde{v}_i > \hat{v}_i$ of data chunk $w_i$.  Because the number $\eta(G)$ of transmissions in Eq.~(\ref{eq:w})  is independent of the valuation set for a fixed coding matrix $G$, we can express function $w(\{\tilde{v}_i,\hat{\mathbf{V}}_{-i}\}, \hat{\mathbf{H}}, G)$ by
	\begin{align}
		w(\{\tilde{v}_i,\hat{\mathbf{V}}_{-i}\}, \hat{\mathbf{H}}, G)=\left\{
		\begin{array}{ll}
		w(\{\hat{v}_i,\hat{\mathbf{V}}_{-i}\}, \hat{\mathbf{H}}, G) + (\tilde{v}_i-\hat{v}_i) & \text{if $G \in \hat{\mathbf{G}}_i$;}\\
		w(\{\hat{v}_i,\hat{\mathbf{V}}_{-i}\}, \hat{\mathbf{H}}, G) & \text{if $G \in \hat{\mathbf{G}}^c_i$,}	
		\end{array}		
	\right.
	\label{eq:w-function-relation}
	\end{align}

	Then, we can obtain
	\begin{align*}
	w(\{\tilde{v}_i,\hat{\mathbf{V}}_{-i}\}, \hat{\mathbf{H}}, G^*)\mathop{=}^{(a)}&w(\{\hat{v}_i,\hat{\mathbf{V}}_{-i}\}, \hat{\mathbf{H}}, G^*) + (\tilde{v}_i-\hat{v}_i)\\
	\mathop{\geq}^{(b)} & w(\{\hat{v}_i,\hat{\mathbf{V}}_{-i}\}, \hat{\mathbf{H}}, G) + (\tilde{v}_i-\hat{v}_i)  \,\,\,\text{for all $G \in  \hat{\mathbf{G}}^c_i$}\\
	\mathop{>}^{(c)}& w(\{\tilde{v}_i,\hat{\mathbf{V}}_{-i}\}, \hat{\mathbf{H}}, G) \,\,\,\text{for all $G \in \hat{\mathbf{G}}^c_i$,}
	\end{align*}  
	where (a) is from Eq.~(\ref{eq:w-function-relation}) along with $G^* \in \hat{\mathbf{G}}_i$; (b) is because coding matrix $G^*$ maximizes function $w(\{\hat{v}_i,\hat{\mathbf{V}}_{-i}\}, \hat{\mathbf{H}}, G)$; (c) is from Eq.~(\ref{eq:w-function-relation}) and $\tilde{v}_i > \hat{v}_i$. Thus, while client $c_i$ submits valuation $\tilde{v}_i$ ($>\hat{v}_i$) of  data chunk $w_i$, the set $\mathbf{G}^*$ of solutions to Eq.~(\ref{eq:vcg-coding}) belongs to $\hat{\mathbf{G}}_i$. That is, the VCG-based coding scheme constructs a coding matrix $G$ such that $\mathbf{1}_i(\hat{H}_i, G)=1$ for sure. Thus, the VCG-based coding scheme is a threshold-type coding scheme.

	\item \textit{The VCG-based payment scheme determines payment $p_i=\bar{v}_i(\hat{\mathbf{V}}_{-i}, \hat{\mathbf{H}})$ if $\mathbf{1}_i(\hat{H}_i, G^*)=1$}: Consider a fixed valuation set $\hat{\mathbf{V}}_{-i}$ and a fixed  side information set $\hat{\mathbf{H}}$. We claim that \textit{threshold $\bar{v}_i(\hat{\mathbf{V}}_{-i}, \hat{\mathbf{H}})$ of the VCG-based coding scheme is}
	\begin{align}
	\bar{v}_i(\hat{\mathbf{V}}_{-i}, \hat{\mathbf{H}})=\max_{G \in \hat{\mathbf{G}}_i^c} w(\{0, \hat{\mathbf{V}}_{-i}\},\hat{\mathbf{H}}, G)-\max_{G \in \hat{\mathbf{G}}_i} w(\{0, \hat{\mathbf{V}}_{-i}\},\hat{\mathbf{H}}, G). \label{eq:threshold}
	\end{align}
	Then, the first term of Eq.~(\ref{eq:vcg-pricing}) 
	\begin{align}
	\max_{G \in \hat{\mathbf{G}}_i \cup \hat{\mathbf{G}}^c_i} w(\{0, \hat{\mathbf{V}}_{-i}\},\hat{\mathbf{H}}, G)=\max_{G \in \hat{\mathbf{G}}^c_i} w(\{0, \hat{\mathbf{V}}_{-i}\},\hat{\mathbf{H}}, G) \label{eq:vcg-pricing-1}
	\end{align}
	is the first term of Eq.~(\ref{eq:threshold}) because of $\hat{v}_i=0$. Moreover, the second term of Eq.~(\ref{eq:vcg-pricing}) 
	\begin{align*}
	w(\{0, \hat{\mathbf{V}}_{-i}\},\hat{\mathbf{H}}, G^*)\mathop{=}^{(a)}&w(\{\hat{v}_i,\hat{\mathbf{V}}_{-i}\},\hat{\mathbf{H}}, G^*)-\hat{v}_i\\
	\mathop{=}^{(b)}&\max_{G \in \hat{\mathbf{G}}_i} w(\{\hat{v}_i,\hat{\mathbf{V}}_{-i}\},\hat{\mathbf{H}}, G)-\hat{v}_i\\
	\mathop{=}^{(c)}&\max_{G \in \hat{\mathbf{G}}_i} w(\{0, \hat{\mathbf{V}}_{-i}\},\hat{\mathbf{H}}, G)
	\end{align*}
	is the second term of Eq.~(\ref{eq:threshold}), where (a) is from Eqs.~(\ref{eq:w-function-relation})  along with $G^* \in \hat{\mathbf{G}}_i$ (because Eq.~(\ref{eq:vcg-pricing}) calculates the price for the case when $\mathbf{1}_i(\hat{H}_i, G^*)=1$); (b) is because $G^*$ maximizes function $w(\{\hat{v}_i,\hat{\mathbf{V}}_{-i}\},\hat{\mathbf{H}}, G^*)$ along with $G^* \in \hat{\mathbf{G}}_i$; (c) is from Eq.~(\ref{eq:w-function-relation}). Then, we complete the proof if the claim is true. 
	
	To establish that claim, we first consider the case when client $c_i$ submits a valuation $\hat{v}_i >\max_{G \in \hat{\mathbf{G}}^c_i} w(\{0, \hat{\mathbf{V}}_{-i}\},\hat{\mathbf{H}}, G)-\max_{G \in \hat{\mathbf{G}}_i} w(\{0, \hat{\mathbf{V}}_{-i}\},\hat{\mathbf{H}}, G)$ of data chunk $w_i$. Then, we can obtain
	\begin{align*}
	\max_{G \in \hat{\mathbf{G}}_i} w(\{\hat{v}_i,\hat{\mathbf{V}}_{-i}\}, \hat{\mathbf{H}}, G)\mathop{=}^{(a)}&\max_{G \in \hat{\mathbf{G}}_i} w(\{0, \hat{\mathbf{V}}_{-i}\}, \hat{\mathbf{H}}, G)+\hat{v}_i\\
	\mathop{>}^{(b)}& \max_{G \in \hat{\mathbf{G}}^c_i}w(\{0, \hat{\mathbf{V}}_{-i}\}, \hat{\mathbf{H}}, G)\\
	\mathop{=}^{(c)}&\max_{G \in \hat{\mathbf{G}}^c_i} w(\{\hat{v}_i,\hat{\mathbf{V}}_{-i}\}, \hat{\mathbf{H}}, G),
	\end{align*}
	where (a) and (c) are from Eq.~(\ref{eq:w-function-relation}); (b) is by the assumption of $\hat{v}_i$. Thus, when $\hat{v}_i> \max_{G \in \hat{\mathbf{G}}^c_i} w(\{0, \hat{\mathbf{V}}_{-i}\},\hat{\mathbf{H}}, G)-\max_{G \in \hat{\mathbf{G}}_i} w(\{0, \hat{\mathbf{V}}_{-i}\},\hat{\mathbf{H}}, G)$, the VCG-based coding scheme constructs a coding matrix~$G$ such that $\mathbf{1}_i(\hat{H}_i, G)=1$. 
	Second, similar to the above argument, when $\hat{v}_i < \max_{G \in \hat{\mathbf{G}}^c_i} w(\{0, \hat{\mathbf{V}}_{-i}\},\hat{\mathbf{H}}, G)-\max_{G \in \hat{\mathbf{G}}_i} w(\{0, \hat{\mathbf{V}}_{-i}\},\hat{\mathbf{H}}, G)$,  the VCG-based coding scheme constructs a coding matrix $G$ such that $\mathbf{1}_i(\hat{H}_i, G)=0$. Fully considering both cases  establishes Eq.~(\ref{eq:threshold}). 

		\item \textit{$\bar{v}_i(\hat{\mathbf{V}}_{-i},\{H_i, \hat{\mathbf{H}}_{-i}\}) \leq \bar{v}_i(\hat{\mathbf{V}}_{-i},\{\hat{H}_i, \hat{\mathbf{H}}_{-i}\})$  under the VCG-based coding scheme for all $\hat{H}_i \subseteq H_i$}: 
Consider a fixed valuation set $\hat{\mathbf{V}}_{-i}$ and a fixed side information set $\hat{\mathbf{H}}_{-i}$. Then, we can obtain  
\begin{align*}
	\bar{v}_i(\hat{\mathbf{V}}_{-i},\{H_i, \hat{\mathbf{H}}_{-i}\})\mathop{=}^{(a)}&\max_{G \in \mathbf{G}} w(\{0, \hat{\mathbf{V}}_{-i}\},\{H_i, \hat{\mathbf{H}}_{-i}\}, G)-\max_{G \in \mathbf{G}_i} w(\{0, \hat{\mathbf{V}}_{-i}\},\{H_i, \hat{\mathbf{H}}_{-i}\}, G)\\
	\mathop{=}^{(b)}&\max_{G \in \mathbf{G}} w(\{0, \hat{\mathbf{V}}_{-i}\},\{\hat{H}_i, \hat{\mathbf{H}}_{-i}\}, G)-\max_{G \in \mathbf{G}_i} w(\{0, \hat{\mathbf{V}}_{-i}\},\{\hat{H}_i, \hat{\mathbf{H}}_{-i}\}, G)\\
	\mathop{\leq}^{(c)} & \max_{G \in \mathbf{G}} w(\{0, \hat{\mathbf{V}}_{-i}\},\{\hat{H}_i, \hat{\mathbf{H}}_{-i}\}, G)-\max_{G \in \hat{\mathbf{G}}_i} w(\{0, \hat{\mathbf{V}}_{-i}\},\{\hat{H}_i, \hat{\mathbf{H}}_{-i}\}, G)\\
	\mathop{=}^{(d)}&\bar{v}_i(\hat{\mathbf{V}}_{-i}, \{\hat{H}_i, \hat{\mathbf{H}}_{-i}\}),
\end{align*}
for all $\hat{H}_i \subseteq H_i$,		where (a) and (d) are from Eqs.~(\ref{eq:threshold}) and (\ref{eq:vcg-pricing-1}); (b) is because $\hat{v}_i=0$; (c) is because $\hat{\mathbf{G}}_i \subseteq \mathbf{G}_i$. 
		\item \textit{If the VCG-based coding scheme in Eq.~(\ref{eq:vcg-coding}) can construct  a coding matrix $G^* \in \mathbf{G}$ such that $\mathbf{1}_i(\hat{H}_i, G^*)=0$ but $\mathbf{1}_i(H_i, G^*)=1$ for some $\hat{v}_i$ and $\hat{H}_i$, then $\bar{v}_i(\hat{\mathbf{V}}_{-i},\{H_i, \hat{\mathbf{H}}_{-i}\})=0$}: Consider a fixed valuation set $\hat{\mathbf{V}}_{-i}$  and a fixed side information set $\hat{\mathbf{H}}_{-i}$.  Suppose that client $c_i$ submits a valuation $\tilde{v}_i>0$ and its complete side information $H_i$. Then, we can obtain 
			\begin{align*}
			&w(\{\tilde{v}_i,\hat{\mathbf{V}}_{-i}\}, \{H_i, \hat{\mathbf{H}}_{-i}\}, G^*)\\
			\mathop{=}^{(a)}&w(\{\hat{v}_i,\hat{\mathbf{V}}_{-i}\}, \{\hat{H}_i, \hat{\mathbf{H}}_{-i}\}, G^*)+\tilde{v}_i\\
			\mathop{\geq}^{(b)} & w(\{\hat{v}_i,\hat{\mathbf{V}}_{-i}\},  \{\hat{H}_i, \hat{\mathbf{H}}_{-i}\}, G)+\tilde{v}_i\,\,\,\text{for all $G\in \mathbf{G}_i^c$}\\
			\mathop{=}^{(c)}& w(\{\hat{v}_i,\hat{\mathbf{V}}_{-i}\}, \{H_i, \hat{\mathbf{H}}_{-i}\}, G)+\tilde{v}_i\,\,\,\text{for all $G\in \mathbf{G}_i^c$}\\
			\mathop{>}^{(d)}&w(\{\hat{v}_i,\hat{\mathbf{V}}_{-i}\}, \{H_i, \hat{\mathbf{H}}_{-i}\}, G)\,\,\,\text{for all $G\in \mathbf{G}_i^c$}\\
			\mathop{=}^{(e)}&w(\{\tilde{v}_i,\hat{\mathbf{V}}_{-i}\}, \{H_i, \hat{\mathbf{H}}_{-i}\}, G)\,\,\,\text{for all $G\in \mathbf{G}_i^c$,}
			\end{align*}
		where (a) is because $\mathbf{1}_i(\hat{H}_i, G^*)=0$ but $\mathbf{1}_i(H_i, G^*)=1$ by assumption; (b) is because $G^*$ maximizes  function $w(\{\hat{v}_i, \hat{\mathbf{V}}_{-i}\}, \{\hat{H}_i, \hat{\mathbf{H}}_{-i}\}, G)$; (c) is from Eq.~(\ref{eq:w-function-relation}) along with $\hat{\mathbf{G}}_i^c \subseteq \mathbf{G}_i^c$; (d) is because $\tilde{v}_i>0$; (e) is from Eq.~(\ref{eq:w-function-relation}). Thus, when $\tilde{v}_i>0$, the VCG-based coding scheme constructs a coding matrix $G$ such that $\mathbf{1}_i(H_i, G)=1$ for sure. Thus, the threshold is $\bar{v}_i(\hat{\mathbf{V}}_{-i},\{H_i, \hat{\mathbf{H}}_{-i}\})=0$.

\end{enumerate}

\section{NP-hardness of Eq.~(\ref{eq:vcg-coding}) in a general set $\mathbf{G}$ of coding matrices}\label{appendix:proposition:general-hard}
We construct a reduction from the original index coding problem, whose objective is to identify a coding matrix for minimizing the number  of transmissions. Given a network instance of the original index coding problem, we construct the same  network instance for our  problem; moreover, set $\hat{v}_i=1$ and $\hat{H}_i=H_i$ for all $i$. In this context, the VCG-based coding scheme can construct a coding matrix $G^*$ such that $\mathbf{1}_i(\hat{H}_i,G^*)=1$ for all $i$,  yielding $w(\hat{\mathbf{V}}, \hat{\mathbf{H}}, G) = n-\eta(G^*)$. Because of the constant $n$, the combinatorial optimization problem in Eq.~(\ref{eq:vcg-coding})  becomes the index coding problem. Then, the NP-hardness of the index coding problem causes the NP-hardness of our problem.

\section{NP-hardness of Eq.~(\ref{eq:vcg-coding}) in the set $\mathbf{G}$ of sparse coding matrices} \label{appendix:lemma:multiple-unicast-hard}
We construct a reduction from the cycle packing problem \cite{cycle-packing}, whose objective is to identify the maximum number of disjoint cycles in a directed graph. Given a directed graph $\mathcal{G}(\mathbf{\Lambda}, \mathbf{A})$ (with  vertex set $\mathbf{\Lambda}$ and  arc set $\mathbf{A}$)  of the cycle packing problem, we construct the weighted dependency graph $\mathcal{G}(\mathbf{\Lambda}, \mathbf{A}, \mathbf{\Gamma})$ with arc weight $\gamma_{a}=1$ for all $a \in \mathbf{A}$. Then, similar to Appendix~\ref{appendix:proposition:general-hard}, we can express function $w(\hat{\mathbf{V}}, \hat{\mathbf{H}}, G)$ in Eq.~(\ref{eq:w-cycle1}) by $w(\hat{\mathbf{V}}, \hat{\mathbf{H}}, G)=n-(n-|\mathbf{C}|)$ for  coding matrix $G$ transmitting along a set $\mathbf{C}$ of disjoint cycles, where the second term $(n-|\mathbf{C}|)$ is the total number $\eta(G)$ of transmissions (because each cycle can save one transmission). Then, the combinatorial optimization problem in Eq.~(\ref{eq:vcg-coding})  becomes the cycle packing problem. The NP-hardness of the cycle packing problem causes the NP-hardness of our problem.

\section{Proof of Theorem \ref{theorem:apx-ratio} }
\label{appendix:theorem:apx-ratio}

Let $\mathbf{C}^*$ be a set of those disjoint cycles  that maximizes the total cycle weight in  weighted dependency graph $\mathcal{G}(\mathbf{\Lambda}, \mathbf{A}, \mathbf{\Gamma})$.
Let graph $\mathcal{G}(\mathbf{\Lambda}_k,\mathbf{A}_k, \mathbf{\Gamma}_k)$ be the remaining graph at the beginning of iteration $k$.  
Let $C_k$ be a cycle minimizing the cycle cost in graph $\mathcal{G}(\mathbf{\Lambda}_k,\mathbf{A}_k, \mathbf{\Gamma}_k)$. Let $\mathbf{C}^*_k$ be the set of those cycles that appear in  set $\mathbf{C}^*$ and also in graph $\mathcal{G}(\mathbf{\Lambda}_k,\mathbf{A}_k, \mathbf{\Gamma}_k)$. By $C_k \cap \mathbf{C}^*_k$ we denote the set of those cycles in $\mathbf{C}^*_k$ that has a common vertex with cycle $C_k$.

We define $APX_k=\gamma(C_k)$ and $OPT_k=\sum_{C \in (C_k \cap \mathbf{C}^*_k)} \gamma(C)$. Because cycle $C_k$ maximizes the cycle weight in iteration $k$, any cycle $C \in C_k \cap \mathbf{C}^*_k$ has cycle weight $\gamma(C)$ less than or equal to $APX_k$. Moreover, since there are $|C_k|$ vertices in cycle $C_k$,  there are at most $|C_k|$  cycles in set $C_k \cap \mathbf{C}^*_k$ (because those cycles are disjoint).  Thus, we can obtain $OPT_k \leq |C_k|\cdot APX_k$. Then, we can complete the proof (by Eq.~(\ref{eq:w-cycle2})) as follow:
\begin{align*}
w(\hat{\mathbf{V}},\hat{\mathbf{H}}, G^*)=&w([\hat{\mathbf{V}}]^+_1,\hat{\mathbf{H}}, G^*)+\sum_{\hat{v}_i\geq 1}(\hat{v}_i-1)\\
\leq&\sum_{k}OPT_k+\sum_{\hat{v}_i\geq 1}(\hat{v}_i-1)\\
\leq& \sum_k |C_k| \cdot APX_k +\sum_{\hat{v}_i\geq 1}(\hat{v}_i-1)\\
\leq& \max_k|C_k| \left(\sum_{k}APX_k+\sum_{\hat{v}_i\geq 1}(\hat{v}_i-1)\right)\\
=&\max_k |C_k|  \left(w([\hat{\mathbf{V}}]^+_1,\hat{\mathbf{H}},G_{\text{alg~\ref{alg:greedy-vcg-coding}}})+\sum_{\hat{v}_i\geq 1}(\hat{v}_i-1)\right)\\
=&\max_k |C_k| \cdot w(\hat{\mathbf{V}},\hat{\mathbf{H}},G_{\text{alg~\ref{alg:greedy-vcg-coding}}}),
\end{align*}
where $[\hat{\mathbf{V}}]^+_1=([\hat{v}_1]^+_1, \cdots, [\hat{v}_n]^+_1)$.

\section{Proof of Theorem \ref{theorme:alg1-truthful} }\label{appendix:alg1-truthful}
We  show that the mechanism consisting of the coding scheme in Alg.~\ref{alg:greedy-vcg-coding} and the payment scheme in Alg.~\ref{alg:greedy-vcg-pricing}   satisfies the four conditions in Theorem \ref{theorem:truthfulness} as follows. 
 
\begin{enumerate}
	\item \textit{Alg.~\ref{alg:greedy-vcg-coding} is a threshold-type coding scheme}: First, suppose that while client $c_i$ submits a valuation $\hat{v}_i$ of data chunk $w_i$,  Alg.~\ref{alg:greedy-vcg-coding} encodes along a cycle $C$ containing vertex $\lambda_i$ in some iteration $k$. Second, suppose that client $c_i$ submits a valuation $\tilde{v}_i >\hat{v}_i$. If Alg.~\ref{alg:greedy-vcg-coding} encodes along a cycle containing vertex $\lambda_i$ before iteration $k$, then we complete the proof; otherwise, if Alg.~\ref{alg:greedy-vcg-coding} cannot, then the cycle $C$ also  maximizes  cycle weight $\gamma(C)$ in iteration $k$ (following Appendix~\ref{appendix:truthful-coding-pricing}). Then, we complete the proof. 
	
	\item \textit{Alg.~\ref{alg:greedy-vcg-pricing} determines payment $p_i=\bar{v}_i(\hat{\mathbf{V}}_{-i}, \hat{\mathbf{H}})$}: Following  Appendix~\ref{appendix:truthful-coding-pricing}, we can establish that $\zeta(C_2)-\zeta(C_1)$ in Line~\ref{alg:greedy-vcg-pricing:s1} of Alg.~\ref{alg:greedy-vcg-pricing} (for an iteration) calculates the minimum valuation $\hat{v}_i$ submitted by client $c_i$ such that  Alg.~\ref{alg:greedy-vcg-coding} encodes along a cycle containing vertex $\lambda_i$ (for that iteration). In particular, in the last iteration, Line~\ref{alg:greedy-vcg-pricing:s1} of Alg. \ref{alg:greedy-vcg-pricing}  produces the minimum valuation $\hat{v}_i$ such that  Alg.~\ref{alg:greedy-vcg-coding} can encode along a cycle containing $\lambda_i$.  Then, we complete the proof. 
	
	\item \textit{$\bar{v}_i(\hat{\mathbf{V}}_{-i},\{H_i, \hat{\mathbf{H}}_{-i}\}) \leq \bar{v}_i(\hat{\mathbf{V}}_{-i},\{\hat{H}_i, \hat{\mathbf{H}}_{-i}\})$  under Alg.~\ref{alg:greedy-vcg-coding}} for all  $\hat{H}_i \subseteq H_i$: Following  Appendix~\ref{appendix:truthful-coding-pricing}, we can establish that the valuation of $\zeta(C_2)-\zeta(C_1)$ for side information set $\{H_i, \hat{\mathbf{H}}_{-i}\}$ is less than or equal to that for side information set $\{\hat{H}_i, \hat{\mathbf{H}}_{-i}\}$ (for each iteration) for all $\hat{H}_i \subseteq H_i$, yielding the result.
	
	\item  Note that when Alg.~\ref{alg:greedy-vcg-coding} encodes coding matrix $G$ along a set of cycles,  we can obtain $\mathbf{1}_i(\hat{H}_i,G)=1$ for clients $c_i$ in the cycles and $\mathbf{1}_i(\hat{H}_i,G)=0$ for clients $c_i$ that are not in those cycles. 	 Moreover, we can also obtain $\mathbf{1}_i(H_i,G)=0$ for each clients $c_i$  that is not in the cycles because no coded data chunks (along the cycles) includes $d_{w_i}$. That is, if $\mathbf{1}_i(\hat{H}_i,G)=0$, then $\mathbf{1}_i(H_i, G)=0$  for all coding matrices $G$ generated by Alg.~\ref{alg:greedy-vcg-coding}. Thus, that case in the fourth condition does not occur.
		

\end{enumerate}

\section{Correctness of Alg.~\ref{alg:max-sqrt-weight}} \label{appendix:lemma:maximum-valuation-cycle}
We prove the lemma by  induction on iteration $i$. When iteration $i=2$, Alg.~\ref{alg:max-sqrt-weight} obviously can produce cycle $C$  for maximizing $\frac{\gamma(C)}{\sqrt{|C|}}$ for $|C| \leq 2$. Next, suppose that when iteration $i=k$,  Alg.~\ref{alg:max-sqrt-weight} produces cycle $C_k$ for maximizing $\frac{\gamma(C)}{\sqrt{|C|}}$ for $|C| \leq k$. Suppose that when iteration $i=k+1$, Line \ref{alg-max-sqrt:fix-i} of Alg.~\ref{alg:max-sqrt-weight} produces cycle $C_{k+1}$  maximizing $\gamma(C)$ for $|C| \leq k+1$. Note that $C_k$ maximizes $\frac{\gamma(C)}{\sqrt{|C|}}$ for $|C| \leq k$; moreover, if $|C_{k+1}|=k+1$, then cycle $C_{k+1}$ maximizes $\frac{\gamma(C)}{\sqrt{|C|}}$ for $|C|=k+1$. Thus, either cycle $C_{k}$ or cycle $C_{k+1}$ maximizes $\frac{\gamma(C)}{\sqrt{|C|}}$ for $|C| \leq k+1$. Thus, Line~\ref{alg-max-sqrt:opt} can produce a cycle $C$ for maximizing $\frac{\gamma(C)}{\sqrt{|C|}}$ for $|C| \leq k+1$. Then, we complete the proof.

\section{Proof of Theorem \ref{theorem:sqrt-approximation} }\label{appendix:theorem:sqrt-approximation}
Follow the notation set in Appendix~\ref{appendix:theorem:apx-ratio}; however, re-define   $C_k$ as a cycle  maximizing $\frac{\gamma(C)}{\sqrt{|C|}}$ in graph $\mathcal{G}(\mathbf{\Lambda}_k,\mathbf{A}_k, \mathbf{\Gamma}_k)$, i.e., $\frac{\gamma(C_k)}{\sqrt{|C_k|}} \geq \frac{\gamma(C)}{\sqrt{|C|}}$ for all cycles $C$ in graph $\mathcal{G}(\mathbf{\Lambda}_k,\mathbf{A}_k, \mathbf{\Gamma}_k)$. Then, we can bound $OPT_k$ above by
\begin{align*}
OPT_k =& \sum_{C \in C_k \cap \mathbf{C}^*_k} \gamma(C) \\
\leq& \frac{\gamma(C_k)}{\sqrt{|C_k|}} \sum_{C \in C_k \cap \mathbf{C}^*_k} \sqrt{|C|} \\
\mathop{\leq}^{(a)}& \frac{APT_k}{\sqrt{|C_k|}} \sqrt{|C_k \cap \mathbf{C}^*_k|} \sqrt{\sum_{C \in C_k \cap \mathbf{C}^*_k} |C|}\\
\mathop{\leq}^{(b)}& \frac{APT_k}{\sqrt{|C_k|}} \sqrt{|C_k|} \sqrt{n}\\
=& \sqrt{n} \cdot APT_k, 
\end{align*}
where (a) is from  the Cauchy-Schwarz inequality; (b) is because  all cycles $C \in C_k \cap \mathbf{C}^*_k$ are disjoint. Hence, $OPT_k \leq \sqrt{n} \cdot APX_k $, yielding the approximation ratio of $\sqrt{n}$.

\section{Hard to approximate Eq.~(\ref{eq:vcg-coding}) in the multiple multicast scenario}\label{appendix:theorem:general-arg-hard-multicast}
We focus on the set $\mathbf{G}$ of sparse and instantly decodable coding matrices. The proof needs the following three technical lemmas.

\begin{lemma} \label{lemma:side-1}
	For the instance of our problem constructed from graph $\mathcal{G}(\mathbf{\Lambda},\mathbf{E})$ of the independent set problem, we can obtain $\mathbf{1}_i(\hat{H}_i, G^*)=1$ for a coding matrix $G^*$ from Eq. (\ref{eq:vcg-coding}).
\end{lemma} 
\begin{proof}
	First, coding matrix $G^*$ can satisfy client $c_{e,1}$ for all $e \in \mathbf{E}$,  because they submit valuations $\hat{v}_{e,1}=1$. Second, suppose that $\mathbf{1}_{e,2}(\hat{H}_{e,2}, G^*)=0$ for some  $e=(x,y) \in \mathbf{E}$. Let  $\mathbf{S}_x=\{c_{e',2}: e' \in \mathbf{E}-\{e\} \mbox{\,\, and is incident to vertex $x$}\}$ be the set of clients (except for client $c_{e,2}$) whose associated edges are incident to vertex $x$. We consider two cases as follows.
	\begin{enumerate}
		\item \textit{$\mathbf{1}_{e',2}(\hat{H}_{e',2}, G^*)=1$ for some $c_{e',2} \in \mathbf{S}_x$}: Because $\mathbf{1}_{e',2}(\hat{H}_{e',2}, G^*)=1$, coding matrix $G^*$ includes the coding vector of  $d_x+d_{e'}$ (such that $\mathbf{1}_{e,2}(\hat{H}_{e,2}, G^*)=0$ and $\mathbf{1}_{e',2}(\hat{H}_{e',2}, G^*)=1$). Then,  substituting the coding vector of  $d_x+d_{e'}$ in coding matrix  $G^*$   by that of   $d_x$ can increase  the function value $w(\hat{\mathbf{V}}, \hat{\mathbf{H}}, G^*)$  by at least $\frac{1}{\text{deg}(x)}$, because the valuation of recovered data chunks increases by $\hat{v}_{e,2}$. That contradicts to the optimality of coding matrix $G^*$. 
		\item \textit{$\mathbf{1}_{e',2}(\hat{H}_{e',2}, G^*)=0$ for all $c_{e',2} \in \mathbf{S}_x$}: Adding the coding vector of  $d_x$ to coding matrix $G^*$ does not change the function value $w(\hat{\mathbf{V}}, \hat{\mathbf{H}}, G^*)$, because the valuation of  recovered data chunks increases by one and the transmission cost also increases by one. 
	\end{enumerate}
Thus, we can obtain $\mathbf{1}_i(\hat{H}_i, G^*)=1$ and  complete the proof.  	
\end{proof}

\begin{lemma} \label{lemma:side-2}
	For the instance of our problem constructed from graph  $\mathcal{G}(\mathbf{\Lambda},\mathbf{E})$ of the independent set problem, we have $\eta^* \leq |\mathbf{E}|+OPT_{vc}$, where $\eta^*$ is the minimum number of transmissions to satisfy all clients, and $OPT_{vc}$ is the minimum size of those vertex covers in graph $\mathcal{G}(\mathbf{\Lambda},\mathbf{E})$. 
\end{lemma}
\begin{proof}
	Let $\mathbf{\Lambda}^* \subseteq \mathbf{\Lambda}$ be a minimum vertex cover in graph $\mathcal{G}(\mathbf{\Lambda}, \mathbf{E})$. We construct a coding matrix $G$ as follows. 
	\begin{itemize}
		\item For each vertex $\lambda \in \mathbf{\Lambda}^*$, add the coding vector of $d_{\lambda}$ to coding matrix $G$;
		\item For each edge $e=(x, y) \in \mathbf{E}$, if both $x \in \mathbf{\Lambda}^*$ and $y \in \mathbf{\Lambda}^*$, add the coding vector of  $d_e$ to coding matrix $G$; if either $x \notin \mathbf{\Lambda}^*$ or $y \notin \mathbf{\Lambda}^*$, add the coding vector of  $d_x+d_e$ or $d_y+d_e$, respectively, to coding matrix $G$. 
	\end{itemize}
	Note that the total number of transmissions made by the constructed coding matrix $G$ is $OPT_{vc}+|\mathbf{E}|$. Moreover, by the following four cases, 
	\begin{itemize}
		\item client $c_{e,1}$ can recover the data chunk it wants with $d_e$, $d_x+d_e$, or  $d_y+d_e$;
		\item for  edge $e=(x, y) \in \mathbf{E}$, if $x \in \mathbf{\Lambda}^*$ and $y \in \mathbf{\Lambda}^*$, then clients $c_{e,2}$ and $c_{e,3}$ can recover the data chunks they want with $d_x$ and $d_y$, respectively;
		\item for each $e=(x, y) \in \mathbf{E}$, if $x \notin \mathbf{\Lambda}^*$ and $y \in \mathbf{\Lambda}^*$, then client $c_{e,2}$ can recover the data chunk it wants with  $d_x+d_e$ and also client $c_{e,3}$ can with $d_{y}$,
		\item for each $e=(x, y) \in \mathbf{E}$, if $x \in \mathbf{\Lambda}^*$ and $y \notin \mathbf{\Lambda}^*$, then client $c_{e,2}$ can recover the data chunk it wants with  $d_x$ and also client $c_{e,3}$ can with $d_{y}+d_e$,
		
	\end{itemize}
	the constructed  coding matrix $G$ can satisfy all clients, yielding $\eta^* \leq |\mathbf{E}|+OPT_{vc}$.
\end{proof}

\begin{lemma} \label{lemma:side-3}
	For the instance of our problem constructed from graph  $\mathcal{G}(\mathbf{\Lambda},\mathbf{E})$ of the independent set problem, we have $\eta^* \geq |\mathbf{E}|+OPT_{vc}$. 
\end{lemma}
\begin{proof}
	First, for satisfying client $c_{e,1}$ (associated with edge $e=(x, y)$) with the instant decoding scheme, the server has to make  at least one transmission (denoted by $t_e$)  of   $d_e$, $d_e+d_x$, or  $d_e+d_y$. Thus, the server has to make at least $|\mathbf{E}|$ transmissions for satisfying client $c_{e,1}$ for all $e \in \mathbf{E}$.


	Second,	let $\tilde{\mathbf{\Lambda}} \subseteq \mathbf{\Lambda}$ be the set of vertices $x$ that have an incident edge $e=(x,y)$ so that client $c_{e,2}$ or $c_{e,3}$ cannot recover the data chunk it wants with transmission $t_e$.  For satisfying  that client $c_{e,2}$, the server has to transmit at least one of  $d_{x}$ or $d_{x}+d_{e}$. Thus,  the server has to make at least another $|\tilde{\mathbf{\Lambda}}|$ transmissions for satisfying all clients in set $\tilde{\mathbf{\Lambda}}$. 
	
	For satisfying all clients, the server has to make at least $|\mathbf{E}|+|\tilde{\mathbf{\Lambda}}|$ transmissions. Note that, for each edge $e\in \mathbf{E}$, at most one of clients $c_{e,2}$ or $c_{e,3}$ can recover the data chunks they want with $t_e$. In particular, for each edge in set $\mathbf{E}$, one of its incident vertices  belongs to set $\tilde{\mathbf{\Lambda}}$. Thus, set $\tilde{\mathbf{\Lambda}}$ is a vertex cover, yielding $|\tilde{\mathbf{\Lambda}}| \geq OPT_{vc}$. To conclude, we can obtain $\eta^* \geq |\mathbf{E}|+|\tilde{\mathbf{\Lambda}}| \geq |\mathbf{E}|+OPT_{vc}$.
\end{proof}

Then, we are ready to prove the result. Because $\mathbf{1}_i(\hat{H}_i, G^*)=1$ for coding matrix $G^*$  (from Lemma \ref{lemma:side-1}),  we can express the function value $w(\hat{\mathbf{V}}, \hat{\mathbf{H}}, G^*)$ in Eq.~(\ref{eq:w}) by 
\begin{align*}
	w(\hat{\mathbf{V}}, \hat{\mathbf{H}}, G^*)=\sum_{e \in \mathbf{E}} \hat{v}_{e,1} + \sum_{e \in \mathbf{E}} (\hat{v}_{e,2} + \hat{v}_{e,3}) -\eta(G^*) =|\mathbf{E}|+|\mathbf{\Lambda}|-\eta(G^*) \mathop{=}^{(a)}|\mathbf{E}|+|\mathbf{\Lambda}|-\eta^*,
\end{align*}
where (a) is because coding matrix $G^*$ (for maximizing $w(\hat{\mathbf{V}}, \hat{\mathbf{H}}, G^*)$) minimizes the number of transmissions. Moreover,   because of $\eta^*=|\mathbf{E}|+OPT_{vc}$ (from Lemmas \ref{lemma:side-2} and \ref{lemma:side-3}), we can obtain $w(\hat{\mathbf{V}}, \hat{\mathbf{H}}, G^*)=|\mathbf{\Lambda}|-OPT_{vc}$.
Let $OPT_{is}$ be the maximum size of those independent sets in graph $\mathcal{G}(\mathbf{\Lambda},\mathbf{E})$. Because of $OPT_{is}+OPT_{vc}=|\mathbf{\Lambda}|$, we finally  obtain $w(\hat{\mathbf{V}}, \hat{\mathbf{H}}, G^*)=OPT_{is}$. Then, the result follows from the hardness of the independent set problem \cite{clq-hard}.

\begin{remark}
To extend the result to the set $\mathbf{G}$ of sparse and general decodable coding matrices, we only need to modify Lemma~\ref{lemma:side-3} as below. First, for satisfying client $c_{e,1}$, the server has to transmit  $t_e=d_e+d$ for some  $d \in \mathbf{D} \cup \emptyset$. Note that if the server transmits  $t_{e}$ and $t_{e'}$  with $t_{e}=t_{e'}=d_{e}+d_{e'}$ for some $e, e' \in \mathbf{E}$, then the server needs to make another transmission for satisfying clients $c_{e, 1}$ and $c_{e',1}$, because both clients do not have data chunks $d_{e'}$ and $d_{e}$, respectively, in their side information. 
	Therefore, for satisfying client $c_{e,1}$ for all $e \in \mathbf{E}$, the server needs to make at least $|\mathbf{E}|$ transmissions. Without loss of generality, we can assume that data chunk $t_e$ is different for all $e \in \mathbf{E}$.
	Second,  for satisfying a client associated with $\lambda \in \tilde{\mathbf{\Lambda}}$ (that wants data $d_{\lambda}$), 
	the server has to make at least one transmission denoted by $t_{\lambda}=d_{\lambda}+d$  for some $d \in \mathbf{D} \cup \emptyset$. If the server transmits $t_{\lambda}=t_{\lambda'}= d_{\lambda}+d_{\lambda'}$ for some $\lambda, \lambda' \in \tilde{\mathbf{\Lambda}}$, then the server needs to make another transmission because both clients (associated with $\lambda$ and $\lambda'$) do not have $d_{\lambda'}$ and $d_{\lambda}$, respectively, in their side information. Thus,  the server has to make at least another $|\tilde{\mathbf{\Lambda}}|$ transmissions for satisfying all clients in set $\tilde{\mathbf{\Lambda}}$. 
	
\end{remark}

\small
\bibliographystyle{IEEEtran}
\bibliography{IEEEabrv,ref}

\end{document}